\documentclass[12pt]{article}
\usepackage{graphicx}                  

\usepackage{xr}
\usepackage[normalem]{ulem}

\usepackage{amsfonts}
\usepackage{array}
\usepackage{amsthm}
\usepackage{color}
\usepackage[usenames,dvipsnames,svgnames,table]{xcolor}
\usepackage{amsmath}
\usepackage{booktabs}
\usepackage{amssymb}
\usepackage{graphicx}
\usepackage{amsfonts}
\usepackage{eurosym} 
\usepackage{float}
\usepackage{array}
\usepackage{amsthm}
\usepackage{amsmath}
\usepackage{appendix}
\usepackage[hang, small, bf]{caption}
\usepackage{cases}
\usepackage{enumerate}
\usepackage{caption}
\usepackage{subcaption}
\usepackage{authblk}
\usepackage{nicefrac}
\usepackage[round]{natbib}

\usepackage{tikz}
\usetikzlibrary{arrows}

\newtheorem{proposition}{Proposition}

\usepackage{amsthm}

\renewcommand{\arraystretch}{1.2}

\usepackage{appendix}

\usepackage{url}

\usepackage{gensymb}

\definecolor{blue}{rgb}{0,0.,0.5}

\makeindex

\begin{document}

\title{\vspace{-2.5cm}\textbf{Rethinking Financial Contagion}
}

\author{Gabriele Visentin,  Stefano Battiston and Marco D'Errico\footnote{Corresponding author: marco.derrico@uzh.ch \\ \textbf{Acknowledgments}: we are grateful to Joseph Stiglitz and Tarik Roukny for discussions and comments. GV, SB and MDE acknowledge support from the FET Project SIMPOL nr. 610704, FET project DOLFINS nr 640772, and the Swiss National Fund Professorship grant no. PP00P1-144689.} \\ \normalsize Department of Banking and Finance \\ University of Zurich
}

\date{\normalsize \today}

\maketitle

\begin{abstract}

\noindent How, and to what extent, does an interconnected financial system endogenously amplify external shocks? This paper attempts to reconcile some apparently different views emerged after the 2008 crisis regarding the nature and the relevance of contagion in financial networks. We develop a common framework  encompassing several network contagion models and show that, regardless of the shock distribution and the network topology, precise ordering relationships on the level of aggregate systemic losses hold among models. 

\noindent
We argue that the extent of contagion crucially depends on the amount of information that each model assumes to be available to market players. Under no uncertainty about the network structure and values of external assets, the well-known \cite{eisenberg2001systemic} model applies, which delivers the lowest level of contagion. This is due to a property of \textit{loss conservation}: aggregate losses after contagion are equal to the losses incurred by those institutions initially hit by a shock.
%
%
This property implies that many contagion analyses rule out by construction any loss amplification, treating \textit{de facto} an interconnected system as a single aggregate entity, where losses are simply mutualised.
Under higher levels of uncertainty, as captured for instance by the DebtRank model, losses become \textit{non-conservative} and get compounded through the network. This has important policy implications: by reducing the levels of uncertainty in times of distress (e.g. by obtaining specific data on the network) policymakers would be able to move towards more conservative scenarios.  

\noindent Empirically, we compare  the magnitude of contagion across models on a sample of the largest European banks during the years 2006- 2016. In particular, we analyse contagion effects as a function of the size of the shock and the type of external assets shocked. 
%
\end{abstract}

\normalsize
\noindent

\clearpage
\tableofcontents
\clearpage
\section{Introduction}
\label{sec:introduction}

Coping with financial crises may carry considerable social costs. Ideally, such costs would need to be correctly quantified and minimised while, at the same time, ensuring that the key functions of the financial system continue to work also under times of distress. Most of the regulator's attention at the onset of the 2007-2008 crisis has been devoted to reducing the probability of widespread default contagion by bailing-out systemically important financial institutions in the fear that they might be ``too big to fail'' (TBTF) or ``too interconnected to fail''
and that their default would lead to severe system-wide consequences. The need to understand the implications of interconnectedness and its consequences in terms of systemic risk has been often emphasized by high-level policymakers in the EU and the US\footnote{\cite{yellen2013intercon} argues that ``some degree of interconnectedness is vital to the functioning of our financial system [\ldots] Yet experience -- most importantly, our recent financial crisis -- as well as a growing body of academic research suggests that interconnections among financial intermediaries are not an unalloyed good
\cite{draghi2013intercon} holds that  ``the process of financial integration has created a myriad of complex linkages within the EU financial system'' and ``a more holistic view of interlinkages in the financial system is needed to understand how shocks are transmitted across the system and how to mitigate them''.} and it has motivated policy responses and substantial revisions of the  regulatory framework.\footnote{For instance, the the Dodd--Frank Wall Street Reform and Consumer Protection Act in the US and the European Market Infrastructure Regulation (EMIR) in the EU.}


Where does this interconnectedness arise? As a matter of fact, a large share of the activities of financial institutions (including banks and other financial players) is \textit{inter-financial}, i.e. it involves contracts between financial institutions themselves. The analysis of \cite{allahrakha2015systemic}, for instance, reports a share of  intrafinancial exposures between about $5\%$ and $40\%$ for the top six banking groups in the US (in our empirical analysis of a sample of EU banks, we find similar numbers). About half these exposures are Over The Counter (OTC) derivatives (i.e. derivative financial contracts traded bilaterally). Several authors \citep[see, for instance,][]{stulz2009credit,ecb2009counterpartycds,cont2010credit} have analysed the role of Credit Default Swaps (CDS) in determining systemic risk. As noted  by the New York FED, ``the financial crisis of 2008 exposed significant weaknesses in the over-the-counter (OTC) derivatives market''.\footnote{\url{https://www.newyorkfed.org/financial-services-and-infrastructure/financial-market-infrastructure-and-reform/over-the-counter-derivatives/index.html}}


The share of inter-financial activities has been increasing over the last decades: \cite{allen1997theory}, for instance, challenge the standard view of the intermediary role of the financial system and argue that these process has led to the creation of ``markets for intermediaries rather than individuals or firms''. 

Is this interconnectedness  a problem in terms of systemic contagion? A first stream of literature focuses on counterfactual simulations to assess the possibility and severity of contagion through national interbank lending markets.\footnote{Among other works, \cite{furfine2003interbank} for the United States, \cite{amundsen2005contagion} for Denmark, \cite{blaavarg2002interbank} and \cite{frisell2007state} for Sweden, \cite{degryse2007interbank} for Belgium, \cite{elsinger2006risk} for Austria, \cite{wells2004financial} and \cite{elsinger2005using} for the United Kingdom, \cite{graf2005interbank} for Mexico, \cite{lubloy2005domino} for Hungary, \cite{mistrulli2011assessing} for Italy, \cite{sheldon1998interbank} and \cite{muller2006interbank} for Switzerland, \cite{toivanen2009financial} for Finland, \cite{upper2004estimating} for Germany and \cite{van2004interbank} for the Netherlands.} All these studies implement models of sequential defaults \citep{upper2011simulation} in which the bankruptcy of one institution reduces the equity of its counterparties by failing to honor its liabilities, either applying the methodologies of  \cite{eisenberg2001systemic} or \cite{furfine2003interbank}.


The general finding of these counterfactual simulations is that the likelihood of domino-effects resulting from interbank exposures is negligible, even though the potential losses associated with the occurence of such events can be of systemic importance \citep{elsinger2006risk,upper2011simulation}. Further extensions of the Eisenberg-Noe model have included bankrupcty costs \citep{rogers2013failure} or cross equity holdings \citep{elsinger2009financial}, without substantially altering the main conclusions.

In their pioneering work, \cite{kiyotaki2002balance} divide contagion into \textit{indirect} and \textit{direct}. The latter refers indeed to a channel arising when ``firms simultaneously borrow
from and lend to each other''. The \textit{indirect} channel has also received attention in years after the crisis: a firm's actions generate externalities that affect other firms through \textit{non-contractual} channels \citep{clerc2016indirect}. Works on fire-sales \citep{cifuentes2005liquidity}, bank runs \citep{iyer2011interbank}, information contagion \citep{helwege2013financial} and liquidity risk \citep{iyer2014interbank} have highlighted the importance of such indirect channels.

%
%

%
%

A stream of literature has suggested that direct contagion may take place \textit{well-before} the threshold of default as in the EN model and lead to substantial distress even without outright defaults. For instance, the DebtRank methodology \citep{battiston2012debtrank} is based on a recursive process on the network of liabilities where shock transmission takes place even without the default of a counterparty, and has been adopted, e.g. in \citep{zlatic2015reduction,de2016evaluating}. As remarked in \citep{battiston2016debtrank,battiston2016leveraging,bardoscia2016paths} DebtRank moves towards extending the classical structural approach of \cite{merton1974pricing} within a network of liabilities. This is precisely modeled in \cite{barucca2016network}.



A reduced likelihood of repayment can be determined by different levels and types of \textit{uncertainty}. Indeed, during the last financial crisis, uncertainty dominated the financial system. As noted by \cite{haldane2009rethinking}, when Lehman filed for bankruptcy, ``Panic ensued. Uncertainty about its causes and contagious consequences brought many financial markets and institutions to a standstill". Other agents in the financial system (whether directly or indirectly exposed) were uncertain about how much they would recover from Lehman. The idea that uncertainty about the recovery rate is crucial for the amount of contagion is put forward also in a recent work by \cite{battiston2016price}: in their model, the levels of uncertainty compounds multiplicatively along chains of connected banks.

International policy organisations have conducted ex-post analyses of the losses emerged during the crisis that support the view that the endogenous mark-to-market re-evaluations of claims, due to the decline in counterparties' ability to repay due obligations, have caused large losses, if not even the biggest part. For instance, \cite{bcbs2011cva} report that 2/3 of losses were due to Credit Valuation Adjustments (CVA) rather than direct defaults;  the \cite{fsa2010prudential}, report a value of 4/5 of losses. CVA implies a re-evaluation of the market value of an asset by taking into account the increased likelihood of default of a counterparty and it is largely applied to OTC derivatives, which are largely interfinancial \citep{derrico2015passing,abad2016shedding}.\footnote{See also \cite{burnham1991structure}, \cite{flood1994market} and \cite{lyons1997simultaneous}, for theoretical and empirical analysis on the FX market.} This shows that a default based process may be insufficient to explain the extent of systemic contagion.\footnote{The theoretical contribution by \cite{glasserman2015likely} find than, under an EN approach contagion is not likely and point towards the inclusion of other mechanisms.} As a result, when balance sheet values are marked-to-market, \textit{valuation} and \textit{contagion} are \textit{strictly interrelated} and may create potential feedback loops and further devaluation spirals.

This leads to the question: does interconnectedness matter in determining contagion?  In order to provide an answer, we develop the a framework based on the concept of \textit{leverage network} (Section \ref{sec:general-framework}) in which we compare (Section \ref{sec:analytical-results}) the losses across five main different contagion models:
\begin{itemize}
\item the \cite{eisenberg2001systemic} model (EN),
\item the \cite{rogers2013failure} model (RV),
\item the Default Cascades (DC) model \citep{battiston2012credit,roukny2013default},
\item the acyclic DebtRank (aDR) model \citep{battiston2012debtrank,battiston2016debtrank,battiston2016leveraging}, 
\item cyclic DebtRank (cDR) model \citep{bardoscia2015debtrank,barucca2016network}.
\end{itemize}
As analytically shown in \citep{barucca2016network}, all these models, including the EN model, can be interpreted as a re-evaluation of interfinancial claims under certain conditions. Moreover, the latter two models explicitly capture a CVA-type of accounting framework.

%
%

As a first contribution, we provide precise ordering relations in terms of total systemic equity loss according to the different models. Importantly, such relations hold \textit{regardless} of i) the underlying network structure and ii) the original shock distribution. In this ordering relations, the EN contagion model gives the least possible losses followed by the RV model and by the cyclic version of DebtRank and the differences in estimated losses between models can be substantial. This implies that the differences in the assumptions of the various models (and in what they are able to capture) is key for assessing the extent of contagion.
%
%
Notice that in the EN and RV models the only event that matters for the transmission of distress from a bank to another is its default. As a result, the value of a claim on a bank's debt is not at all affected by a loss of equity of that bank as long as the equity is positive, while in case of default of the bank the value of the claim for their counterparties is, recursively, a proportional share of all the remaining assets of the bank including what the bank could recover from its own defaulting obligors. 
In particular, losses that exceed the equity of a bank are mutualized by the defaulting bank' counterparties as in a sort of system of communicating vessels where the water in excess of one vessel moves to the next but the total amount of water remains constant. 

Our second contribution is to emphasise that an important \textit{conservation} property holds for the EN model:  the total final losses suffered by banks and their creditors at the end of the contagion process are simply equal to the total initial losses on the external assets of the banks suffering the first shock. This property, which has often gone unnoticed, implies that the possibility of a shock amplification through the network is ruled out by construction and that the interbank network is equivalent in terms of aggregate losses to a single aggregate bank. 
The other implicit assumptions of the EN model are related to the certainty in how much each counterparty will recover from the others (there are no bankruptcy costs and assets are liquidated at book value) and how payments are enforced. Moving towards more uncertainty on the network structure and on the recovery rate makes the contagion process no longer conservative: the final total losses in the system are higher than the initial shock. In the RV model, for example, losses are computed recursively taking into account a discount factor for the interbank liabilities. This factors compounds in the network, lowering the final recovery rate. Yet, the RV model assumes full knowledge on the network of liabilities and external asset values. 
When further uncertainty on external and interfinancial assets (i.e., the network) ensues, a full mark-to-market/CVA must be applied and an approach such as DR is more suitable. This approach leads to even higher levels of losses.


As a third (empirical) contribution, in Section \ref{sec:empirical}, we apply the stress-test framework in a cross-national setting, studying the vulnerability of the top $50$ publicly listed banks on the EU interbank lending market from $2005-Q4$ to $2015-Q3$.\footnote{Though pure interbank loans are not subject to a mark-to-market reevaluation, we focus on interbank loans as a benchmark case to assess the magnitude of the loss.} We also compare the behavior of the different models across several dimensions in the space of parameters, such as the magnitude of external shock and exogenous recovery rates, identifying the limits under which different models yield similar results. Our empirical results reflect the theoretical findings: sequential default models (like the Eisenberg-Noe model) constantly yield lower equity losses than mark-to-market-based models.
%


Finally, we discuss (Section \ref{sec:discussion}) the implications of our results, especially from a policy perspective. In fact, our paper provides theoretical ground for policy actions aimed at mitigating systemic risk in times of crises. Increasing information about the network structure and other holdings of market participants is available, policymakers may move the financial system towards a more conservative setting by curbing uncertainty. We discuss also the types of policies that may lead to these specific desirable results. However, further work is needed to understand their actual implementation in the current complex financial landscape.

\section{A general framework for contagion models}
\label{sec:general-framework}

\subsection{Leverage networks}
\label{subsec:leverage-networks}

In order to compare different contagion models we hereby provide a coherent, unified framework of definitions. In what follows we will outline such a framework and review all the models analyzed.

Consider a set of $n$ financial institutions (hereafter called \textit{banks}) with nominal liabilities to each other.  Let $L^b$ be the \textit{interbank liability matrix}, whose entry $L_{ij}^b$ represents the nominal value of the liability of bank $i$ to bank $j$. It is worth to stress that the liability $L_{ij}^b$ can represent any type inter-financial contract, including OTC derivatives (cleared and uncleared), bonds, etc.: our theoretical analysis, therefore, is not limited to interbank loans.\footnote{The current framework does not apply straightforwardly, though, in the presence of a specific type of financial network with CDS contracts written on the very same financial institutions, as in \cite{schuldenzuckerclearing}.}

We can represent this financial system as a directed graph (or network) with weighted adjacency matrix given by the interbank liabilities matrix $L^b$, in which each bank is a node. Associated with the matrix $L^b$ is its transpose, which we call the \textit{interbank assets matrix} and denote by $A^b$, whose entry $A_{ij}^b = L_{ji}^b$ represents the nominal value of bank $i$'s claim on bank $j$. Banks invest also in $m$ \textit{external assets}, such that their total amounts to $A_i^e = \sum_{k=1}^m A^e_{ik}$, where $A_{ik}^e$ is the amount invested by bank $i$ in asset class $k$; moreover, they have \textit{external liabilities}, $L^e_i$, that is liabilities to entities \textit{outside} of the inter-financial network. Denoting by $E_i$ the equity of $i$, the balance sheet identity for each node in the network is therefore given by the following expression:
\begin{equation}
 E_i = A_i - L_i = A_i^e + A_i^b - (L_i^e + L_i^b) = \sum_{k=1}^m A_{ik}^e + \sum_{j=1}^n A_{ij}^b - L_i^e - \sum_{j=1}^n L_{ij}^b \label{eq:balancesheetidentity}
 \end{equation}
One of the key quantities in our framework is the \textit{leverage} of bank $i$, defined as
\begin{equation}
l_i =A_i \big/ E_i. \label{eq:totalleverage}
\end{equation}

The leverage ratio can be interpreted \citep{battiston2016leveraging} as the multiplicative factor that relates a decrease in the value of the assets with the losses in equity it induces. In fact, given a negative relative shock, $s \in [0, 1]$, in the value of the assets, the fraction of equity needed to respect identity (\ref{eq:balancesheetidentity}) is readily computed as $\min\{1, l_i \times s\}$. 

Another simple interpretation is that the leverage ratio can be viewed as the reciprocal of the minimum relative decrease (negative shock) in the value of assets that leads to $E_i = 0$, i.e. $i$'s equity is completely used in order to absorb the shock. \footnote{See, e.g., \cite{gros2010leverage}, for a discussion about this interpretation of the leverage ratio within different accounting frameworks.}

The framework is  based on the decomposition of the leverage ratio into \textit{additive components}, as introduced in \citep{battiston2016debtrank,battiston2016leveraging}. The leverage ratio of $i$ \textit{with respect to asset}  $k$ is defined as the ratio between the value of the asset and the equity of the bank, $l_{ik}^e = A_{ik}^e / E_i$. The same applies to interbank exposures, and $l_{ij}^b$ is the leverage of $i$ towards $j$. Therefore the total leverage $l_i$ can be additively decomposed along external assets and interbank assets in the following way:

\begin{equation}
\label{eq:decomposition}
 l_i = \frac{\sum_{k=1}^m A_{ik}^e + \sum_{j=1}^n A_{ij}^b}{E_i} =  \sum_{k=1}^m l_{ik}^e + \sum_{j=1}^n l_{ij}^b 
\end{equation}

We call the matrices whose elements are $l_{ik}^e$ and $l_{ij}^b$ the \textit{external leverage matrix} and the \textit{interbank leverage matrix} respectively. To these matrices are naturally associated a bipartite leverage network banks/assets and a monopartite leverage network banks/banks.


\subsection{First and second round effects}
\label{subsec:first-second-round-effects}

In general, given a model of financial contagion, we are interested in assessing how the losses induced by a shock, $s$, on the external assets of some banks (\textit{first round effects}) propagate and amplify to the other banks in the network (\textit{second round effects}).

We show that all contagion models considered in this work can be equivalently reformulated as discrete-time recursive processes\footnote{Even models framed as fixed-point problems, like the Eisenberg-Noe and Rogers-Veraart models ultimately admit a discrete-time process solution, like the Fictitious Default Algorithm.} on the leverage network. Each process prescribes how a certain balance sheet quantity (e.g. total liabilities or equity) evolves, compatibly with accounting rules and key assumptions related to each model, up until convergence is reached. 

Without loss of generality and independently of the details of any model implementation, we can define a time scale\footnote{Time here refers to the internal time of the process itself, or algorithmic time, and not physical time in any way.} for this process of \textit{stress-testing} (compatibly with the general framework proposed in \citet{battiston2016leveraging}) as in Table \ref{tab:outlinedistress}:
%

\begin{table}[!htbp]\index{typefaces!sizes}
  \footnotesize%
  \begin{center}
    \begin{tabular}{lll}
      \toprule
\textbf{Time} &\begin{tabular}{l} \textbf{Round} \end{tabular}& \begin{tabular}{l} \textbf{Effects} \end{tabular}\\
      \midrule
\begin{tabular}{l}$t = 0$ \\ $\;$ \end{tabular}       &  \begin{tabular}{l}  Initial allocation \\ $\;$ \end{tabular}&\begin{tabular}{l} Initial conditions for model, resources are  \\allocated according to nominal values \end{tabular}\\
& & \\\midrule
\begin{tabular}{l}$t = 1$ \\ $\;$ \end{tabular}&  \begin{tabular}{l} First round \\ $\;$ \end{tabular}& \begin{tabular}{l} Shock on external assets, loss on balance sheets \\ $\;$  \end{tabular} \\ \midrule
\begin{tabular}{l}$t = 2$ \\ $\;$ \end{tabular}&  \begin{tabular}{l} Second round begins \\ $\;$ \end{tabular} & \begin{tabular}{l}Reverberation of first round losses on the interbank network \\  according to the model's  dynamics\end{tabular}\\
& & \\\midrule
\begin{tabular}{l}$t = \infty$ \end{tabular}& \begin{tabular}{l} Second round ends  \end{tabular}   & \begin{tabular}{l}Model reaches convergence \end{tabular}\\ 
& & \\
      \bottomrule
    \end{tabular}
  \end{center}
  \vspace{-0.5cm}
  \caption{Outline of distress process in the general framework.}
  \label{tab:outlinedistress}
\end{table}

As a convenient placeholder we will indicate with $t= \infty$ the time of convergence of any model, even in those cases in which the model reaches convergence in finite time.\footnote{We further notice that all the models considered can equivalently be re-phrased in terms of fixed-points, so that at convergence if we proceed iterating the dynamics, we will always remain at the fixed-point, without further changes in the system}

The focus of the generalised framework of this paper is on the following two main quantities:

\begin{enumerate}
\item the \textit{individual vulnerability} of bank $i$
\begin{equation}
\label{eq:individual-vulnerability}
h_i(t) = \frac{E_i(0) - E_i(t)}{E_i(0)} \in [0,1]
\end{equation}
defined as the relative cumulative equity loss of bank $i$ up to time $t$,
\item the \textit{global vulnerability} of the system
\begin{equation}
\label{eq:global-vulnerability}
 H(t) = \frac{E_{tot}(0) - E_{tot}(t)}{E_{tot}(0)} = \sum_{i=1}^n \frac{E_i(0)}{\sum_{j=1}^n E_j(0)} h_i(t) \in [0, 1]
 \end{equation}
defined as the relative cumulative equity loss of the system up to time $t$.
\end{enumerate}

Clearly $h_i(t) = 1$ if bank $i$ has defaulted\footnote{Compatibly with modern literature, we assume that a default event is triggered when the bank is insolvent, ie. its liabilities are greater than its assets. For such a defaulted bank we define its equity to be exactly zero, $E_i = 0$, in accordance with the principle of limited liability} at any time up to time $t$. Further, we assume that for the entire duration of the stress-test no ``injection'' of equity is performed, so that $h_i(t)$ is non-decreasing in $t$	, i.e. $h_i(t) \geq h_i(t-1)$. 

These quantities have the obvious advantage of providing the monetary value of losses in equity incurred by the system (upon multiplying by initial equity) and are compatible with the principle of limited liability, allowing us to compare different models in a straightforward way.

\section{Conservation versus amplification of losses}
\label{sec:analytical-results}

We now explore how contagion can lead to more or less shock amplification in the financial system discussing the analytical results contained in  the \ref{sec:backgr-inform-models}. We show that, when a sufficient amount of information on the network is available, so that the EN model can be implemented, the system is \textit{conservative}, i.e. initial losses are not amplified. Increasing the uncertainty leads to larger losses and less conservative scenarios. 

\subsection{Partial ordering}
\label{subsec:partial-ordering}

Within the leverage network framework, we prove that, independently of the network structure and shock distribution, there exists a general partial ordering among the five different models in terms of individual and global vulnerabilities.

Preliminary to the analysis, we show in Proposition \ref{prop:EN-RV-leverage} that the EN and RV models can be re-expressed in terms of individual vulnerabilities, thus allowing a direct comparison across the models. A general overview of the models' main assumptions and features is outlined in Table \ref{tab:model-summary}. 

\begin{table}[!htbp]\index{typefaces!sizes}
  \footnotesize%
  \begin{center}
    \begin{tabular}{lll}
      \toprule
\textbf{Model} &\begin{tabular}{l} \textbf{Set of active nodes} \end{tabular}& \begin{tabular}{l} \textbf{Contagion mechanism} \end{tabular}\\
      \midrule
\begin{tabular}{l}EN \\ $\;$ \end{tabular}       &  \begin{tabular}{l}  Defaulted banks \\ $h_i(t) = 1$ $\;$ \end{tabular}&\begin{tabular}{l} Counterparty default losses \\ Endogenous recovery rate\end{tabular} \\ \midrule
\begin{tabular}{l} RV \\ $\;$ \end{tabular}       &  \begin{tabular}{l}  Defaulted banks \\ $h_i(t) = 1$ $\;$ \end{tabular}&\begin{tabular}{l} Counterparty default losses \\ Endogenous and \\ exogenous recovery rate\end{tabular}\\ \midrule
\begin{tabular}{l}DC \\ $\;$ \end{tabular}&  \begin{tabular}{l} Defaulted banks \\ $h_i(t) = 1 \mbox{ and } h_i(t-1) < 1$ $\;$ \end{tabular}& \begin{tabular}{l} Counterparty default losses \\ Exogenous recovery rate \\ (propagation only once) $\;$  \end{tabular} \\ \midrule
\begin{tabular}{l}aDR \\ $\;$ \end{tabular}&  \begin{tabular}{l} Distressed banks \\ $h_i(t) > 0 \mbox{ and } h_i(t-1) = 0$ $\;$ \end{tabular} & \begin{tabular}{l} Counterparty default and mark-to-market losses \\ Exogenous recovery rate \\ (propagation only once)\end{tabular}\\ \midrule
\begin{tabular}{l}cDR \end{tabular}& \begin{tabular}{l} Distressed banks \\ $h_i(t) > 0$ \end{tabular}   & \begin{tabular}{l}Counterparty default and mark-to-market losses \\ Exogenous recovery rate \end{tabular}\\ \bottomrule
    \end{tabular}
  \end{center}
  \vspace{-0.5cm}
  \caption{Summary of properties of distress propagation models.}
  \label{tab:model-summary}
\end{table}

Then,  we prove that the following chain of inequalities is always satisfied (see Propositions \ref{prop:EN<RV}, \ref{prop:RV<cDR}):

\begin{equation}
H^{\text{EN}}(t) \leq H^{\text{RV}}(t) \leq H^{\text{cDR}}(t).
\end{equation}

This result asserts that models that allow distress propagation only in the event of defaults, such as the EN and RV models, always yield lower vulnerability values than particular models incorporating MtM losses, thus corroborating the empirical evidence (see Introduction) that the greatest part of counterparty risk losses are determined by mark-to-market accounting more than outright defaults.

We further establish that for the case of the aDR model no such relationship holds and that in general

\begin{equation}
H^{\text{EN}}(t) \nleq H^{\text{aDR}}(t), \quad H^{\text{DC}}(t) \nleq H^{\text{aDR}}(t).
\end{equation}

This is a direct consequence of the fact that the aDR model, despite implementing MtM losses, allows banks to propagate distress only once. This feature can be exploited to build pathological counterexamples in which the EN, RV and DC model predict higher vulnerabilities (see Propositions \ref{prop:not-DC<aDR}, \ref{prop:not-EN<aDR}). Nevertheless it is important to emphasize that such a situation requires a certain degree of fine-tuning in the choice of parameters and underlying network structure, resulting in a rather unrealistic financial system.

We remark that in the case of all reconstructed financial networks analyzed in our empirical application, the following stable order relationship among models is satisfied (for instance, see Figure \ref{fig:time-external-assets}):

\begin{equation}
H^{\text{EN}}(t) \leq H^{\text{DC}}(t) \leq H^{\text{RV}}(t) \leq H^{\text{aDR}}(t) \leq H^{\text{cDR}}(t).
\end{equation}

As outlined in the Introduction, a large body of work on contagion is based on the EN model, which systematically yields very low second round effects, followed by the other models, where the interplay between the recovery rate and increasing uncertainty yields to increasingly higher losses. 

Why does the EN lead to very low losses? The answer lies in a specific property of this model, i.e. that initial losses are \textit{conserved}, as we shall see in the following.

\subsection{Conservation of losses in EN}
\label{subsec:conservation-losses}

An important aspect of the EN model is that it was originally conceived \citep{eisenberg2001systemic} as a \textit{clearing} algorithm, attempting to find a solution for the problem of payments among a set of interconnected firms, with a set of endogenous endowments that are less than the nominal value of the liabilities in the system (in our framework this is captured by an initial exogenous shock). Their main claim is that not only a solution for the clearing problem exists, but that - under mild assumptions - it is unique.\footnote{The solution to the clearing problem can be found either via a fixed-point argument or via a linear program. Using the latter approach, for example, \cite{liu2010sensitivity} provides a detailed sensitivity analysis of the problem.} 

Indeed, looking more in depth, the EN method requires a series of precise conditions that need to be met in order to be implemented.  These conditions represent implicit assumptions for the application of the model in a real setting. Overall, they lead to reduced sources of uncertainty and allow to reach an optimal clearing in a network of interfinancial exposures. Among others, some of the assumptions are:

\begin{enumerate}
\item External shocks are fully known and fixed in time;  \cite{barucca2016network} show that, when external shocks follow a stochastic process, instead, the valuation has to be closer to the one performed in the DebtRank model;
\item Complete knowledge of banks' balance sheets, which implies further,
\item the complete knowledge of the network; which would have eliminated a large source of uncertainty in the system during the financial crisis as described by \cite{haldane2009rethinking}.
\item All agents abide to the payment solution proposed to them given by the clearing algorithm (no uncertainty in the payments out and in) which would rule out uncertainty due to the possibility of defection. Should one of the participants refuse to pay the amount, the whole clearing process should be repeated, leading to further possible defections.
\item Institutions have full recovery of external assets, thus excluding any fire sales and liquidity costs.
\item There are no bankruptcy costs, no default would lead to a (potentially lengthy) process of unwinding of positions.
\item Even a \textit{complete default} (no remaining equity) is not enough for a node $i$ to propagate distress onto its neighbours: the shortfall in payments due by $i$ needs to exceed its equity level in order to be spread to other nodes.
\item At default, all remaining assets are \textit{liquidated immediately} and with certainty. This implies that counterparties should not react to any changes (however large) in the equity levels of their obligor and patiently wait until all external and interbank assets are used (this is in contrast with current accounting practice and with the \cite{merton1974pricing} approach).
\end{enumerate}

In particular, complete knowledge of interbank exposures and full compliance to the clearing system imply the existence of an omniscent, omnipotent central agent able to collect data and fully enforce payments (via binding legal agreements), guaranteeing optimal market clearing. Whether or not this can be achieved in a non-centralised way remains an open problem.

Under the global clearing mechanism implemented in the EN model, losses are conserved and fully mutualized among banks, thus leading to negligible second-round contagion effects (Propositions \ref{prop:EN-final-H}, \ref{prop:EN-second-round}, and \ref{prop:EN-second-round-upper-bound}). This allows us to provide an explicit closed-form

\begin{equation} 
\label{eq:en-at-convergence}
H^{\text{EN}}(\infty) = \frac{1}{\sum_{i=1}^n E_i(0)} \sum_{i=1}^n \bigg( A^{e}_i s_i - (1 - \beta_i) (\bar{p}_i - p_i(\infty)) \bigg)
\end{equation}

where $s_i$ is the shock on the external assets of bank $i$ and $\beta_i$ is its \textit{financial connectivity}, i.e. the fraction of $i's$ liabilities held by other banks. The term $\bar{p}_i - p_i(\infty)$ represents the \textit{total shortfall in payments} of $i$. In the summation of Equation (\ref{eq:en-at-convergence}) the term $A_i^e s_i$ represents the loss in external assets of bank $i$, while the term $(1 - \beta_i) (\bar{p}_i - p_i^*)$ (which is non-zero only if bank $i$ defaults) represents that portion of external liabilities that are written-off. Clearly, if $\beta_i=1, \forall i$, the total equity losses are equal to the total initial shocks on external assets, while second round losses are simply equal to those losses in external assets that are in excess of the equity of the banks defaulting during the first round.  If $\beta_i<1$, part of the shortfall in payments is borne by external debtholders (e.g. depositors). 

In a sense, the two terms have opposite effects: the shortfalls in payments that external debtholders (for instance, depositors) are forced to endure represent financial distress that flows \textit{out} of the system (hence the minus sign in front of the second term), thus reducing the second round losses that are passed onto counterparties of defaulting institutions.

We can  obtain the following upper bound on second round losses (see Proposition \ref{prop:EN-second-round-upper-bound}):

\begin{equation}
 H^{EN}(\infty) - H^{EN}(1)  \leq \frac{1}{\sum_{i=1}^n E_i(0)}  \sum_{i \in \mathcal{D}(1)} \beta_i \bigg( A^{e}_i s_i - E_i(0) \bigg). \nonumber
 \end{equation}

where $\mathcal{D}(1)$ is the set of banks that default as a consequence of the first round shock on external assets. 

We see that losses due to contagion are solely due to those first round losses that defaulting banks' equity cannot absorb at the end of the first round.\footnote{Indeed, in the original paper, \cite{eisenberg2001systemic} had provided this very intuition ``\textit{the financial system is \textit{conservative}, neither creating nor destroying value}''. Within our framework, this is explicitly proven in \ref{subsec:si-cons-loss-en}.}

For instance, if bank $i$ defaults as a consequence of first round losses, then a fraction $1 - \beta_i$ of losses in excess of its equity flows out of the network, as discussed above, while a fraction $\beta_i$ is passed onto its counterparties in the network and alone contributed to second round effects. The network then acts as a system of communicating vessels: bank $i$'s counterparties try to absorb these losses until their equity is exhausted and, if a default is triggered, they proceed to write-off external liabilities (thus determining further out-flow of distress) and interbank liabilities (thus passing distress onto their own counterparties). 

The process ends when all those initial losses have been re-distributed along the chain of counterparties and/or externalized onto depositors. Clearly, the worst-case scenario for the financial system is that in which $\beta_i = 1,\: \: \forall i$, i.e. banks do not have outside liabilities and have to internalize all shocks. Interbank losses are then maximized and the conservation law is even clearer and reads:

\begin{equation}
\label{eq:en-at-convergence-agg}
\begin{array}{ccccc}
H^{\text{EN}}(\infty) \times \sum_{i = 1}^n E_i(0) &=& \left[\sum_i E_i(0) - \sum_i E_i(\infty) \right] &=&
\\
& = & \left[\sum_i \left(E_i(0) - E_i(\infty)\right)\right]&= & \sum_i s_i A_i^e 
\end{array}
\end{equation}

It can be readily observed that equity losses in the system are thus exactly equal to the losses in assets determined by the initial external shock. The clearing mechanism simply re-distributes distress in the network in such a way that no amplification is observed. 

Dividing Equation (\ref{eq:en-at-convergence-agg}) by the original total equity in the system, one obtains (assuming a common shock $s_i = s, \: \: \forall i$):

\begin{equation}
H^{\text{EN}}(\infty) = 1 - \frac{\sum_i E_i(\infty)}{\sum_i E_i(0)} = s \frac{\sum_i A_i^e}{\sum_i E_i} = s l_{sys}^e
\end{equation}

Thus the final equity loss in the system is independent of the underlying network topology (see examples in Section \ref{subsec:si-illustrative-examples}) and is solely a function of the external shock and the total leverage of the system. 

Therefore, the implementation of the EN clearing mechanism is equivalent to aggregating all the balance sheets of the banks and computing first round losses on the aggregated entity: the structure and arrangement of the internal connections between individual balance sheets is irrelevant for the computation of the final losses ultimately shouldered by the financial system. From the point of view of a policymaker quantifying the amount of equity injection in the system during a crisis, all that matters is the original shock and the fraction of external debt. 

This means that, despite the original formulation as a recursive process on networks, in the EN model the banking system acts as a single bank with an aggregate balance sheet and satisfies the principle of conservation of losses. Therefore the algorithm is needed soley for the correct individual attribution of losses, which would incentivize participants to adopt the clearing algorithm, rather than for an aggregate assessment.

\section{Empirical application\label{sec:empirical}}

We apply the stress-test framework presented in Section \ref{sec:general-framework} in a cross-national setting, studying the vulnerability of the top $50$ publicly listed banks on the EU interbank lending market according to the different models. We have collected quarterly balance sheet data for these banks from 2005-Q4 to 2015-Q3.\footnote{Details concerning data collection and processing can be found in Section \ref{subsec:si-data-collection}}

For each quarter in our dataset we proceed to reconstruct $1000$ realizations of the network of interbank exposures via a modified fitness model \citep{demasi2006fitness,musmeci2012bootstrapping,cimini2015systemic,
cimini2015estimating,battiston2016leveraging}, as explained in more detail in Section \ref{subsec:si-network-reconstruction}. For each quarter, we run the five different models\footnote{For background information on the models, see Section \ref{subsec:si-contagion-models}} on every realization of the $1000$ interbank networks and compare the resulting median global vulnerabilities. 

We explore the behavior of the different models across several dimensions in the space of parameters, in particular we study how global vulnerabilities are affected by:
\begin{itemize}
\item different asset classes (all external assets, derivatives and impaired loans)
\item varying magnitude of external shocks
\item varying exogenous recovery rates
\end{itemize} 
and we identify the limits under which different models yield similar results.

\subsection{Vulnerability in time across asset classes}
\label{sec:time}

As a starting point we compute first round, $H(1)$, and second round, $H(\infty)$, losses\footnote{Notice that the value of $H(1)$, depending only on the initial balance sheets and external shock, is identical for all models. Only  second round effects, quantified as $H(\infty)$, are different, since they're the result of  the particular dynamics of the individual contagion processes} for a $1\%$ fixed shock on the external assets of each bank and compare models across the whole period 2005-Q4 - 2015-Q3, on a quarterly basis.

\begin{figure}[H]
\centering
\includegraphics[width=\columnwidth]{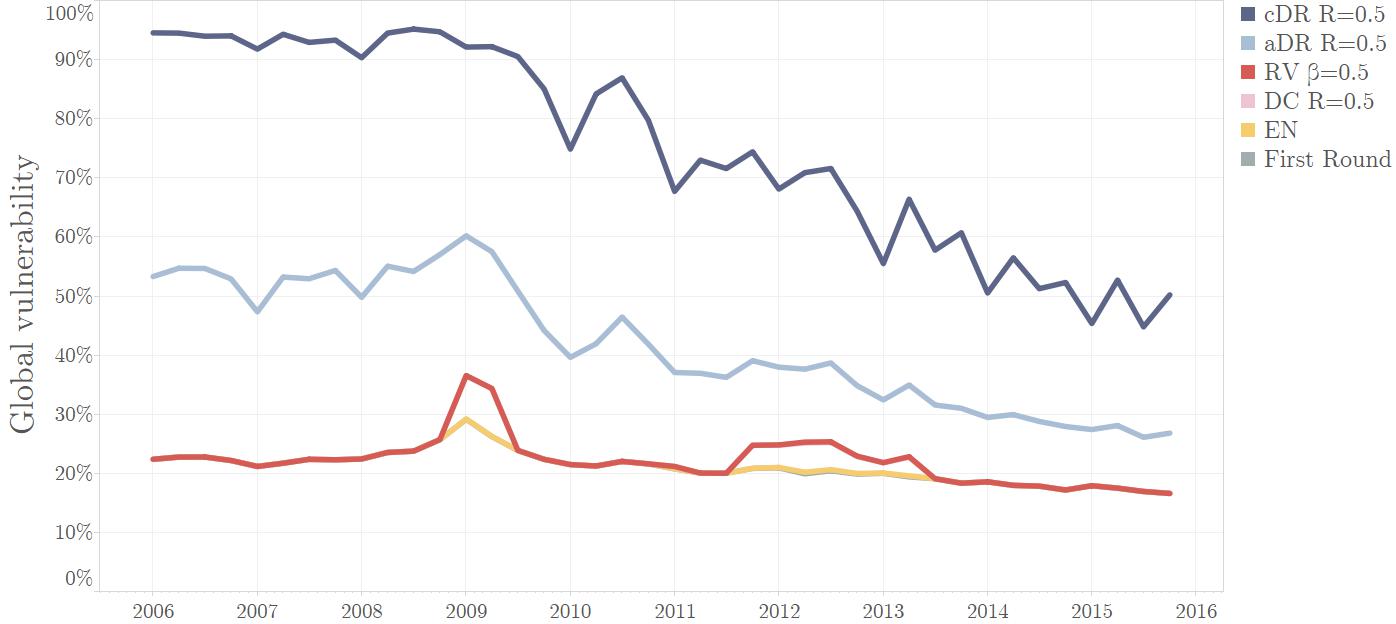}
\caption{Global vulnerability, $H$, in time for a 1\% shock on external assets. ``First round'' corresponds to $H(1)$, identical across all models. Second round, $H(\infty)$, for each of the five models.}
\label{fig:time-external-assets}
\end{figure}

Figure \ref{fig:time-external-assets} presents the results of such a simulation. The global vulnerability for the EN model, $H^{EN}(\infty)$, is almost identical to the first round and covers it entirely in the picture, while the DC model lies underneath the RV model.

We notice that the aDR and cDR models constantly result in higher levels of vulnerability, while the DC model yields only slightly smaller values than the RV model. That this is always the case is demonstrated analytically in Section \ref{subsubsec:si-relations-among-models}, together with several other relations among the models.

This, in particular, implies that the models explicitly integrating the practice of mark-to-market accounting are able to take into account losses due to CVA that the standard EN and RV model simply cannot capture.

It is also evident that the EN model fails to produce any visible second round effect, even though the banking system in that period faced its biggest financial crisis since 1929. This is the result of several assumptions underlying the EN model, as illustrated in Section \ref{subsec:partial-ordering} and proved analytically in Section \ref{subsec:si-cons-loss-en}, that indicate that the EN model is not suited for stress-testing. In this regard it is particularly worrying that most practitioners used this model in reaching the conclusion that financial contagion through the interbank market is unlikely. Our findings prove that the EN model, created as a clearing payment system, was indeed designed to minimize this kind of losses and all the assumptions on which it's based are clearly non satisfied in a context of stress-testing.

It is worth emphasizing that, nevertheless, the RV model, thanks to the implementation of an exogenous recovery rate, is able to capture a portion of the equity losses at times when financial distress was at its peak (namely the first quarters of 2009, when EU banks registered the highest losses in market capitalization, and the 2012-2013 sovereign debt and Greek crises). But these second round effects are limited to the periods of greatest strain in the system and can simply be used to identify periods of full-fledged crisis when it is well under way. In contrast, the aDR and cDR models provide visible trends in second round effects, clearly showing a build-up of vulnerability in the years preceeding the crisis and a generalized easing of conditions later on. In particular, as indicated by \cite{battiston2016financial} and \cite{bardoscia2016paths}, the values for cDR can be conveniently related to the first eigenvalue of the leverage matrix.

\begin{figure}[!htbp]
\centering
\includegraphics[width=\columnwidth]{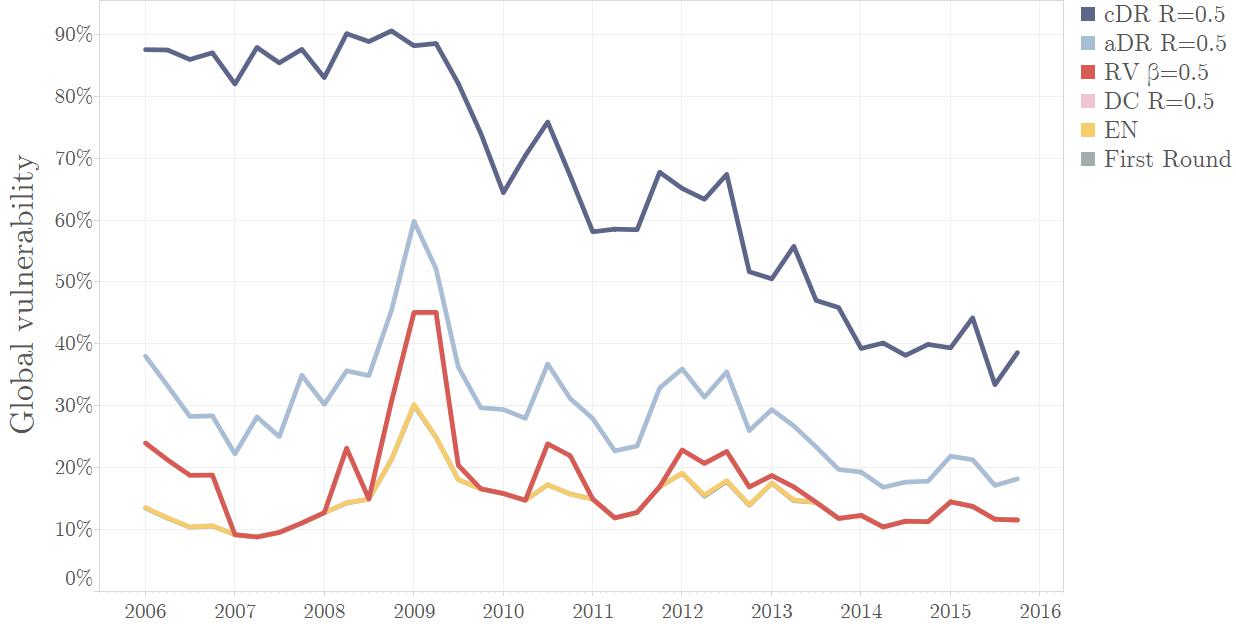}
\caption{Global vulnerability, $H$, in time for a 5\% shock on all derivatives. ``First round'' corresponds to $H(1)$, identical across all models. Second round, $H(\infty)$, for each of the five models.}
\label{fig:time-derivatives}
\end{figure}

\begin{figure}[H]
\centering
\includegraphics[width=\columnwidth]{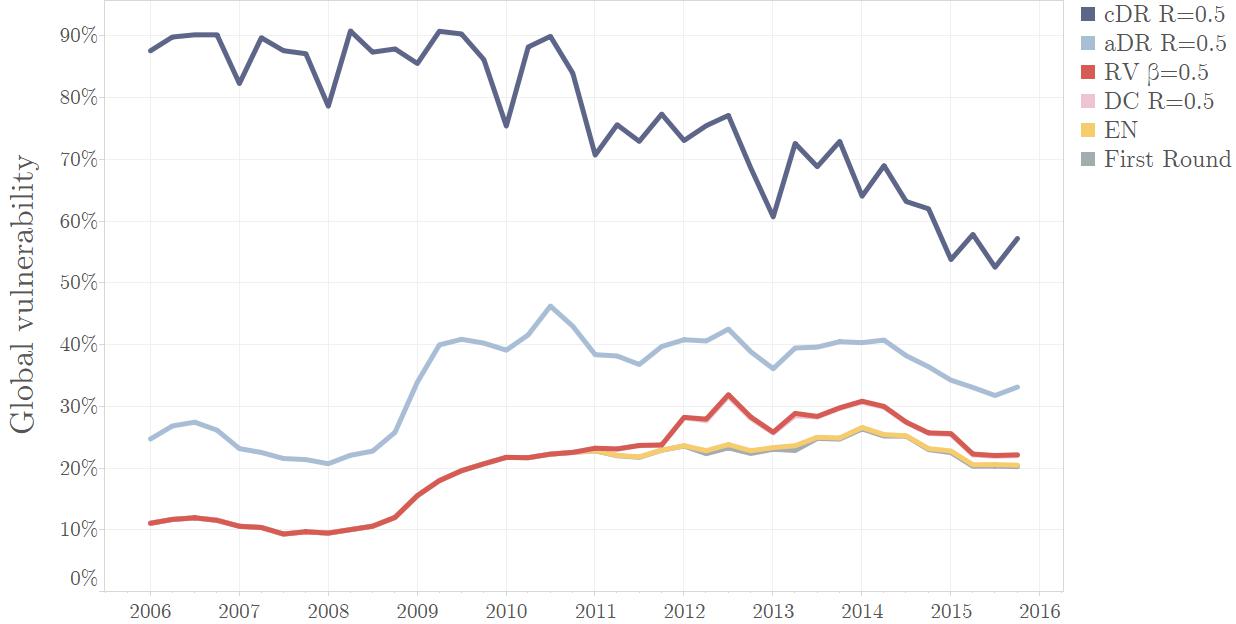}
\caption{Global vulnerability, $H$, in time for a 40\% shock on all impaired loans. ``First round'' corresponds to $H(1)$, identical across all models. Second round, $H(\infty)$, for each of the five models.}
\label{fig:time-impaired-loans}
\end{figure}

Figures \ref{fig:time-derivatives} and \ref{fig:time-impaired-loans} show exactly the same results for shocks\footnote{We do not seek to justify this particular choice of shocks on derivatives and impaired loans, since this is meant to be a comparative analysis of different distress models and not an actual stress-test. The magnitude of these shocks has therefore been chosen in such a way that first round losses are roughly comparable with the $1\%$ shock presented in Figure \ref{fig:time-external-assets}} on different asset classes, namely a shock of $7\%$ on all derivatives and a shock of $40\%$ on all impaired loans\footnote{Banks have considerable discretion in how they classify impaired loans and there is no commonly accepted definition. A loan is usually defined impaired when payments are 90+ days overdue, but in our dataset, as in the case of all data providers, the figure of impaired loans might be undereported}.

The decomposition of contagion through particular asset classes allows us to identify the assets that most contribute to fragility in the system and to understand how systemic risk migrates from one typology of investment to another.

Derivatives, for example, appear to have contributed strongly to the vulnerability of banks during the financial crisis because they were particularly leveraged with respect to this kind of instruments.

Impaired loans, on the other hand, seem to pose a systemic risk mainly in times of recession,\footnote{It must be emphasized nevertheless how the widespread securitisation of loans, especially sub-standard ones, might have contributed substantially to the financial crisis, but due to the fact that these assets were registered off the balance sheets of most banks, we couldn't capture them in our data} when sluggish growth in the real economy implies higher default probabilities on retail, residential and commercial loans.

\subsection{Vulnerability for different shocks}
\label{sec:shocks}

\begin{figure}[H]
\centering
\includegraphics[width=\columnwidth]{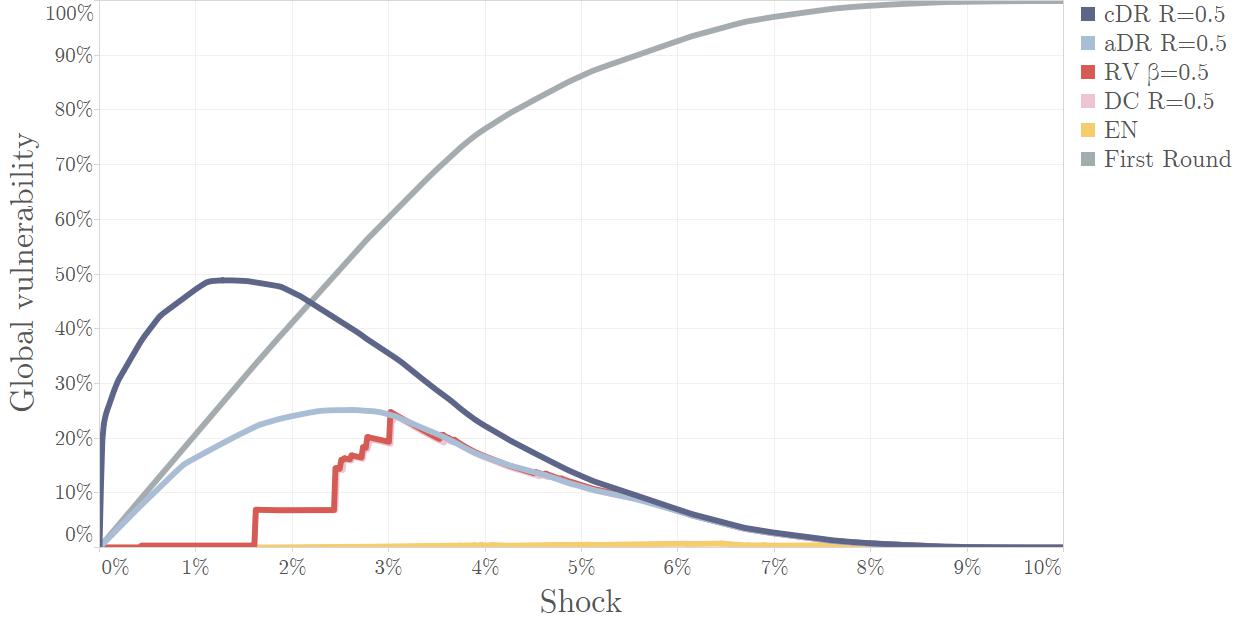}
\caption{Global vulnerability, $H$, as a function of shock on external assets. ``First round'' corresponds to $H(1)$, identical across all models. Second round, $H(\infty)$, for each of the five models. Simulation on a unique realization of the interbank network in 2010-Q2.}
\label{fig:shocks-H}
\end{figure}

The dependence of second round effects on the size of the initial shock is another important aspect of contagion models. 

In Figure \ref{fig:shocks-H} we plotted first round effects and second round effects for the aDR, RV, DC and EN models as functions of shock on all external assets in 2010-Q2 for a given realization of the interbank network in that quarter.  It can be seen how shocks and second round effects are, quite intuitively, not trivially proportional, in fact as the shock is increased the fraction of losses incurred during the first round becomes predominant, literally leaving smaller and smaller amounts of equity for the second round to deplete.

For small shocks ($s < 1.5\%$) the RV, EN and DC models capture no significant second round losses, because the shock is so small that virtually no bank is in default and therefore no bank can propagate distress. The aDR and cDR on the other hand, accounting for CVA, show second round losses due to mark-to-market accounting even for those shocks. 
The relationship between fraction of defaulted banks and shocks is presented for the very same simulation in Figure \ref{fig:shocks-defaulted}, where it can be seen that it takes a certain fraction of first round defaults to trigger second round losses under the EN, RV and DC models.

\begin{figure}[H]
\centering
\includegraphics[width=\columnwidth]{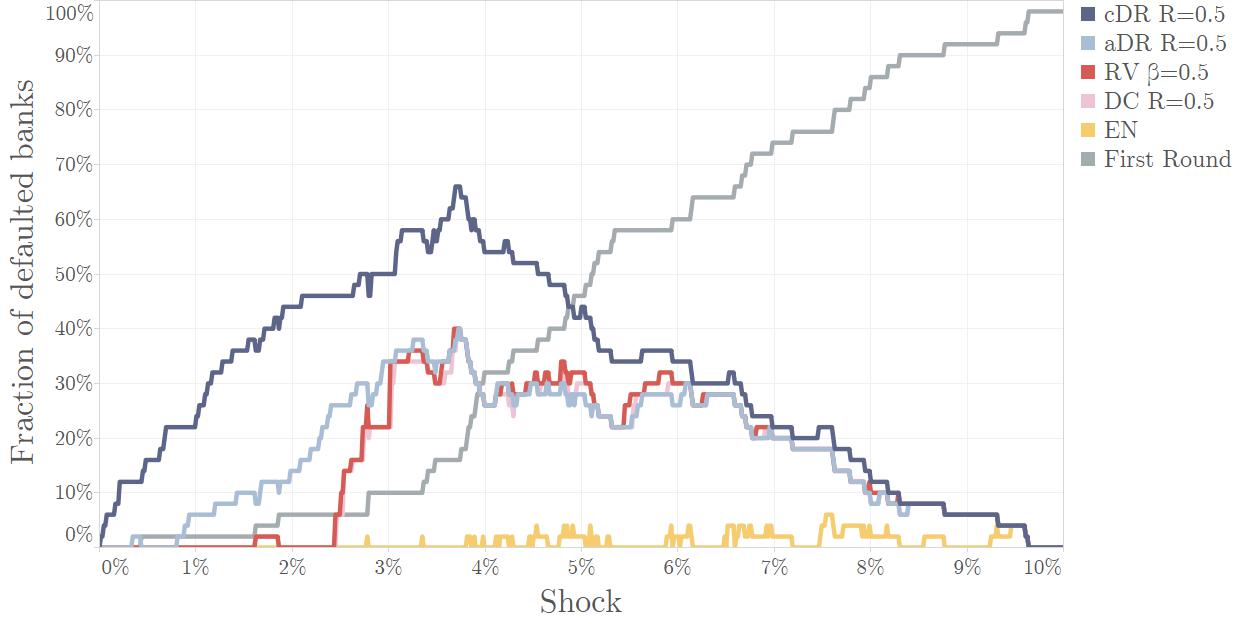}
\caption{Fraction of defaulted banks in the system as a function of shock on external assets. ``First round'' corresponds to fraction of banks defaulted under the external shock and is identical across all models. Fraction of defaulted banks at second round is shown for each of the five models. Simulation on a unique realization of the interbank network in 2010-Q2.}
\label{fig:shocks-defaulted}
\end{figure}

For medium sized shocks ($ 1.5\% < s < 3\%$) the second round effects of every model, except aDR and cDR, are strongly network-dependent. Losses are triggered depending on whether or not certain well-connected institutions are shocked and different shocks imply different defaults at first round. This in particular implies that to compute expected monetary losses from a stress-test based on these models with reasonable accuracy, full and exact knowledge of the underlying interbank network must be assumed. 

For large shocks ($s > 3\%$) the models converge roughly to a common estimate of second round losses, due to the fact that the majority of banks have, by now, defaulted at the first round and are all able to propagate financial distress to the few solvent counterparties left.

\subsection{Vulnerability for different recovery rates}
\label{sec:recovery-rates}

Mark-to-market contagion models, like aDR and cDR, tend to yield substantially higher vulnerabilities. An important connection between mark-to-market models and the other ones is given by the recovery rate. 

In Figure \ref{fig:recovery-rates-H} we present results similar to those in Figure \ref{fig:shocks-H}, but for three distinct values of the recovery rate and only for the particularly interesting cases of the aDR and RV models.
As can be seen the models tend to converge to the same distribution of second round effects for sufficiently large shocks, remaining substantially different only for small shocks. 

\begin{figure}[!htbp]
\centering
\includegraphics[width=\columnwidth]{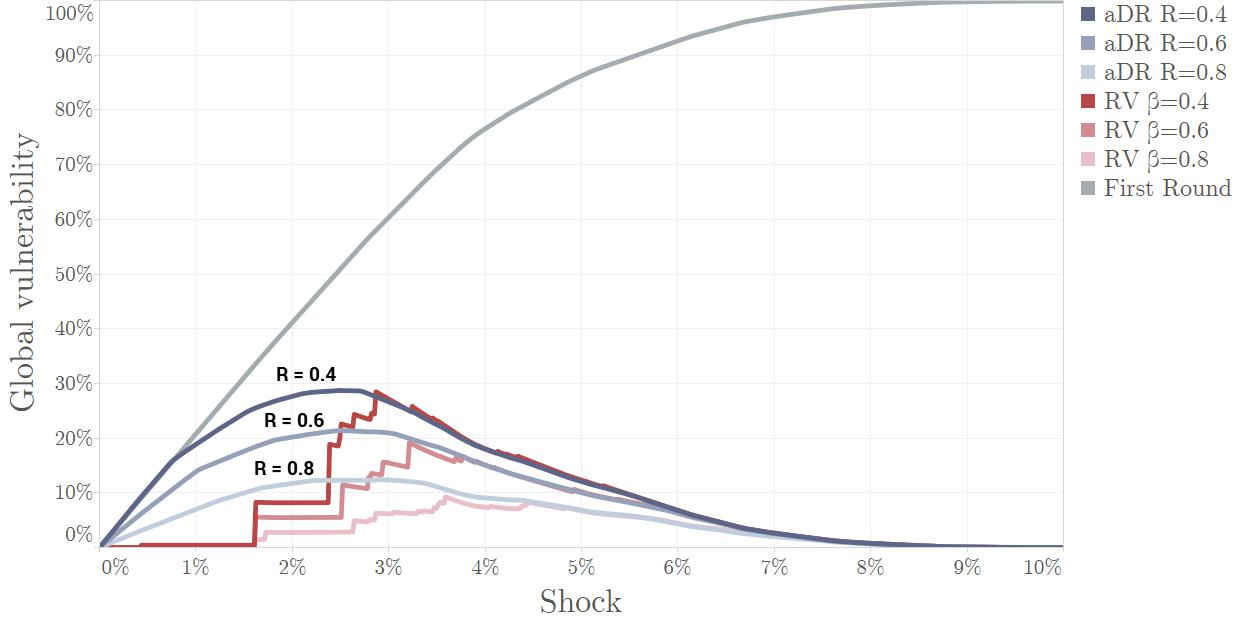}
\caption{Global vulnerability as a function of shock on external assets. First round and second round for aDR and RV for increasing values of the recovery rate. Simulation on a unique realization of the interbank network in 2010-Q2.}
\label{fig:recovery-rates-H}
\end{figure}

\begin{figure}[!htbp]
\centering
\includegraphics[width=\columnwidth]{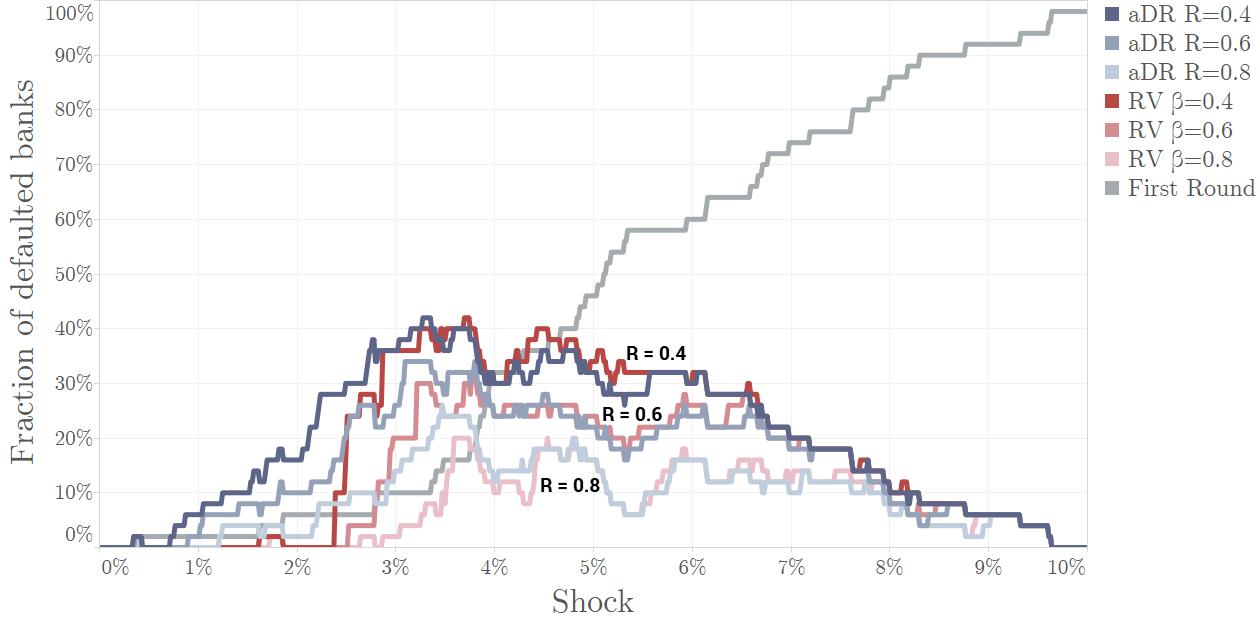}
\caption{Fraction of defaulted banks as a function of shock on external assets. First round and second round for aDR and RV for increasing values of the recovery rate. Simulation on a unique realization of the interbank network in 2010-Q2.}
\label{fig:recovery-rates-defaulted}
\end{figure}

\section{Discussion and policy relevance\label{sec:discussion}}

By introducing a common framework (based on the notion of \textit{leverage network}) for several contagion models, we prove analytical results about the amount of contagion in a set of financial institutions with interconnected liabilities. In particular, the ordering relationship we find is suitable to compare models as it holds \textit{regardless} of the network structure.

Overall, our results point towards the conclusion that the level of systemic losses arising from interconnectedness crucially depends on the amount of information about the network structure, coupled with the expected recovery rate on the liabilities. We find that systemic effects are increasing in the levels of uncertainty and exacerbated by an interconnected system.

We show that the well-known Eisenberg and Noe contagion model is \textit{conservative}, in the sense that initial losses are not amplified by the network of liabilities but merely redistributed. Adopting this type of approach to ``clear'' liabilities during a crisis would be then desirable in order to reduce systemic losses to a minimum level. As losses would be in great part limited and fully absorbed by individual equity cushions, this would put into a new perspective both the notions of ``too big to fail'' and ``too interconnected to fail''.

However, as we have discussed, for the EN approach to be applied, a series of conditions need to be satisfied. In particular, it requires full knowledge of the network of liabilities. These specific condition was not present, for example, when the 2007-2008 crisis ensued. As \cite{haldane2009rethinking} points out, the financial system came to a ``standstill'' after Lehman's default, as its position in the CDS market ``was believed'' to be large. Back then, this implied two levels of uncertainty:  the first is about the value of external assets, the second is about the network of exposures, thus making almost impossible to implement the EN approach in that setting. 

This leads to an important policy implication of our work. As an EN-like model is more desirable to curb systemic losses, how could it be  applied in a real setting?

First, more data at the level of individual positions should be made available to policymakers. The increase in data reporting is indeed the road taken by recent reforms such as the European Market Infrastructure Regulation (EMIR) for OTC derivatives in the European Union\footnote{\url{http://eur-lex.europa.eu/legal-content/EN/TXT/?uri=CELEX:32012R0648}} and the Dodd Frank Act in the US. Policymakers should then investigate the key aspects of the financial network, including mapping ``risk flows'' \citep{derrico2015passing} and the ``geography or risk''.\footnote{As the Governor of the Central bank of Ireland Philip R. Lane has very recently pointed out in a speech \url{http://www.bis.org/review/r160610b.pdf}.}

Second, encouraging  the possibility of setting up specific institutional arrangements or post-trade infrastructures, under these conditions, would greatly help in applying to its fullest the EN algorithm and limit the losses. These arrangements would need to be ideally legally binding (at least to a certain extent). In this situation, the system would be able to move towards more and more conservative losses. In the presence of the correct type and amount of data, the regulator or the infrastructure would be able to enforce almost instantaneously the payments due by each participant in order to minimise losses.  
If correctly put in place (despite the natural operational difficulties), this policy (data availability) can prove to be helpful for mitigation. This operation could also include market players directly:  Nobel laureate E. Fama has recently suggested that, during the crisis,  market participants could have even ``unscrambled everything'' in ``a week or two''.\footnote{\url{http://www.newyorker.com/news/john-cassidy/interview-with-eugene-fama}} As a more long-term objective, it would be even possible to run a \textit{real-time clearing} of the financial system that would help policymakers to understand the build up of systemic distress.


However, our work provides only theoretical support and not yet an operative solution on how to achieve this. Further work is necessary to bring the theoretical finding to a fully operative level. We are aware that a complete conservation of losses may be not fully achieved when the maturity structure and the complexity of financial products grows: micro and macroprudential policies should continue taking such complexity into account. Efforts in this direction include pre-trade activities such as central clearing \citep{abad2016shedding} and post-trades practices such as compression \citep{derrico2016compression}, which are becoming increasingly adopted by market participants, and leading to a more structured system.





\clearpage

\appendix
\renewcommand{\thesection}{SI}
\section{Supplementary Information}
\label{sec:backgr-inform-models}

\subsection{Five contagion models}
\label{subsec:si-contagion-models}

We now briefly present the five contagion models analyzed in the paper within the general stress-testing framework presented in Section \ref{subsec:leverage-networks}.

\subsubsection{Eisenberg-Noe model (EN)}
\label{subsubsec:si-en}

The Eisenberg-Noe model was first introduced in \cite{eisenberg2001systemic} as a clearing payment mechanism for a network of financial institutions in which some nodes had defaulted. The EN model addresses this issue by determining the payments institutions have to make in order to minimize individual and collective losses, under the assumptions of absence of bankrutpcy costs and liquidation costs. As such it is not a model specifically designed for stress-testing, but has been used extensively in recent literature (see \cite{elsinger2006risk}, \cite{glasserman2015likely}, among many others) exactly for this purpose. We briefly review the main definitions and results, connecting them with our general framework formulation.

In the same context of the financial system described in Section \ref{sec:general-framework}, define the \textit{total obligations vector} $\bar{p}$, with components 
\begin{equation}
\bar{p}_i = \sum_{j=1}^n L^b_{ij} + L^e_i \nonumber
\end{equation}
as the total liabilities of bank $i$, both internal and external to the network, with equal seniority.\footnote{In the original formulation of the model \citep{eisenberg2001systemic} banks are assumed to have no external liabilities. However, a system in which $\beta_i = 1 \forall i$ is incompatible with a system for all banks have leverage strictly larger than one: only-lenders banks must have unit leverage. However, banks are naturally deposit-collecting institutions from the real economy. So having a fraction $1- \beta_i > 0$  of deposits over total liabilities represents a more realistic choice. Hence, we present here a reformulation of the model used in virtually all recent works in the field, that is the natural extension of the old model in the presence of external liabilities and proportionality of payments. For more details, see \cite{elsinger2013models} and for an application see \cite{glasserman2015likely}.}
Let \begin{equation}
\Pi_{ij} = 
\begin{cases}
L_{ij}/\bar{p}_i \quad & \mbox{ if $\bar{p}_i > 0$} \\
0 & \mbox{ otherwise}
\end{cases}
\end{equation}
be the \textit{relative liabilities matrix}. This matrix represents the relative value of nominal liabilities of bank $i$ to bank $j$ as a proportion of the debtor's total liabilities. Generally $\Pi$ is a row sub-stochastic matrix, unless external liabilities are zero, in which case $\Pi$ is row stochastic. Define $A_i^e$ as the external assets of bank $i$, so that the \textit{net worth} of a bank is given by
\begin{equation}
\label{eq:EN-net-worth}
\tilde{E}_i(\bar{p}) = \sum_{j=1}^n \Pi_{ij}^T \bar{p}_j + A_i^e - \bar{p}_i
\end{equation}
This expression corresponds to the book-value of the equity of the bank, $E_i$, only in case $\tilde{E}_i(\bar{p})$ is non-negative. If it is strictly negative, then it indicates that the bank's assets are less than its liabilities, which means that the bank is \textit{insolvent} and cannot honor its liabilities completely and creditors will therefore necessarily accept debt write-offs.

Let us now introduce the \textit{payment vector}, $p$, representing all payments of bank $i$ to debtholders, both internal and external to the network. In general not all payments vectors will be feasible, in the sense that for an arbitrary choice banks might be required to pay more than they hold in assets or less than  their liabilities.

\cite{eisenberg2001systemic} have proved, under fairly general conditions, the existence and uniqueness of a \emph{clearing payment vector}, i.e. a particular payment vector that satisfies the basic accounting principles of: i) limited liabitility, ii) absolute priority, and iii) proportionality of claims; that is
\begin{equation}
\label{eq:en-clearing-payment-vector}
p_i^* = \min \bigg\{ \sum_{j=1}^n \Pi_{ij}^T p_j^* + A_i^e, \bar{p}_i \bigg\}.
\end{equation}
Such a clearing payment vector can be found through a discrete-time algorithm, called Fictitious Default Algorithm. This algorithm, starting from the candidate payment vector $\bar{p}$, produces a sequence of plausible payment vectors, $p(t)$, until convergence is reached in maximum $n$ (i.e. the number of nodes in the network) iterations. At convergence, the clearing payment vector $p^* = p(\infty)$ thus yields to:
\begin{equation}
E_i(\infty) =  
\begin{cases}
0 \quad & \mbox{if $p_i = \bar{p}_i$} \\
\sum_{j=1}^n \Pi_{ij}^T p_j(\infty) + A_i^e - p_i(\infty) \quad & \mbox{otherwise.}
\end{cases}
\end{equation}

Notice that at convergence the equity of defaulted banks is necessarily zero. From now on we will always indicate the clearing payment vector with the symbol $ p(\infty)$, to stress that we will work solely with the discrete-time Fictitious Default Algorithm implementation of the model, in agreement with the remarks made in Section \ref{subsec:first-second-round-effects}.

A key quantity in our framework is the \emph{endogenous recovery rate}, defined as 
\begin{equation}
\label{eq:en-recovery-rate}
p_j(\infty) / \bar{p}_j \in [0, 1]
\end{equation}

which determines the fraction of nominal value of liabilities that bank $j$ can effectively repay to its creditors once the system is cleared. For non-defaulting banks this value is always $1$.

The algorithm can be used in stress-testing, as we do, by ``shocking'' the external assets of banks by a given fraction, $s_i$, so that the equity at $t=1$ is

$$ \tilde{E}_i(1) = \max \bigg\{0, \sum_{j=1}^n \Pi_{ij}^T \bar{p}_j + (1 - s_i)A_i^e - \bar{p}_i \bigg\}$$

and from $t=2$ the Fictitious Default Algorithm starts on the shocked balance sheets, so that at the $t$-th step the equities read
\begin{equation}
\label{eq:en-equities-time-t}
\tilde{E}_i(t) = \max \bigg\{0, \sum_{j=1}^n \Pi_{ij}^T p_j(t) + (1 - s_i)A_i^e - p_i(t) \bigg\},
\end{equation}

\subsubsection{Rogers-Veraart model (RV)}
\label{subsubsec:si-rv}
The Rogers-Veraart model \citep{rogers2013failure} is a generalization of the Eisenberg-Noe model that takes into account bankruptcy and liquidation costs. Equation \ref{eq:en-clearing-payment-vector} can be modified to account for these effects in the following way:

\begin{equation}
\label{eq:rv-clearing-payment-vector}
p_i^* = \min \bigg\{ \beta \sum_{j=1}^n \Pi_{ij}^T p_j^* + \alpha A_i^e, \bar{p}_i \bigg\},
\end{equation}

where $\alpha$ is a coefficient quantifying the fraction of nominal value of the external assets recovered after liquidation costs (recovery rate on external assets) and $\beta$ quantifies the fraction of nominal value of interbank assets that the defaulted bank can recover after bankruptcy costs or, alternatively, corresponds to a discount factor on the liquidated interbank claim due to early settlement.

Rogers and Veraart show that the results of existence and uniqueness of the clearing payment vector are true also for this extension of the model and a straightforward variation of the Fictitious Default Algorithm allows efficient computation in exactly the same fashion.

We will always deal with the case $\alpha = \beta$ and denote this common value\footnote{This assumption leads to substantial improvements in mathematical tractability and for this reason it is also made, for example, in \cite{glasserman2015likely}.} by $\beta$. This parameter $\beta$ can be interpreted as an additional multiplicative factor that compounds the endogeneous recovery rate obtained through the algorithm, the latter being computed as in Equation \ref{eq:en-recovery-rate}.

\subsubsection{Default Cascades model (DC)}
\label{subsubsec:si-dc}

An alternative approach to the one proposed by Eisenberg and Noe is that of modelling the propagation of losses from defaulted banks to counterparties directly in terms of a contagion process on the equity of banks. This approach was followed in \cite{battiston2016price} and \cite{battiston2012credit}.

Using the definitions of vulnerability outlined in Section \ref{sec:general-framework}, we define recursively a discrete-time dynamic process on the graph given by the following equation:
\begin{equation}
 h_i(t+1) = \min \bigg\{1, h_i(t) + \sum_{j \in \mathcal{A}(t)} (1 - R)l_{ij}^b h_j(t) \bigg\}, \nonumber
 \end{equation}
where $l_{ij}^b = A_{ij}^b / E_i(0)$ is the interbank leverage of bank $i$ towards bank $j$, $R$ is the exogenous recovery rate, and $\mathcal{A}(t)$ is the \textit{set of active nodes} at time $t$, that is the set of nodes that can propagate distress at time $t+1$, defined as
\begin{equation}
\mathcal{A}(t) = \{ j \:\:\: | \:\:\: h_j(t) = 1 \mbox{ and } h_j(t') < 1, \quad \forall t'<t \}.\nonumber
\end{equation}

The process is easily explained as follows. When bank $j$ defaults, that is $h_j(t) = 1$, it reduces its payments to counterparty $i$ of a value $(1 - R)A_{ij}^b$. This loss in asset value is absorbed by the equity of its counterparty $i$ and determines a relative equity loss of magnitude $(1 - R) A_{ij}^b / E_i(0)$. If the total losses incurred by $i$ up to time $t+1$ exceed the value of equity at $t=0$, then $h_i(t+1) = 1$ and bank $i$, having exhausted its shock buffer against losses, is itself in default. The peculiar definition of the set of active nodes ensures that a node imparts losses to counterparties: i) only once, and ii) exactly at the time step immediately following its default. 

Once a bank has defaulted, its vulnerability will remain equal to one in all successive iterations. Since losses are propagated on the simple directed graph induced by the interbank nominal liabilities matrix and cycles are ignored, convergence to $h(\infty)$ is reached after at most $\text{diam}(L)$ iterations.\footnote{The diameter $\text{diam}(L)$ of the interbank network, is the maximum length of all shortest paths in the network between any pair of nodes}

\subsubsection{Acyclic DebtRank model (aDR)}
\label{subsubsec:si-adr}

All the models seen so far assume that losses are propagated exclusively by the default event of a counterparty. In other words they assume that only defaults can determine a loss in asset values on the balance sheets of counterparties.

Nevertheless, the current practice requires banks to implement a mark-to-market evaluation of their current claims against their respective counterparties. In the derivative market, for instance, this goes under the name of Credit Valuation Adjustment (CVA).\footnote{Introduced by the \cite{ifsa2006statement} for US banks before the crisis and subsequently internationally sanctioned by the \cite{basel2010global}, CVA is an instance of Fair Value Accounting (FVA), in that it involves reporting assets and liabilities on the balance sheet at fair value and recognizing changes in fair value as gains and losses in the income statement. All regulatory implementations require the use of mark-to-market (MtM) values for fair value accounting, whenever available. In particular, CVA takes into account counterparty and credit risk by requiring to explicitly account for the deterioration in creditworthiness of counterparties (according to some plausible models).  See, for instance, the BCBS Consultative Document on CVA at \url{http://www.bis.org/bcbs/publ/d325.pdf} or the more recent EBA report at \url{https://www.eba.europa.eu/documents/10180/950548/EBA+Report+on+CVA.pdf}.
Since its early implementation CVA has attracted numerous detractors \citep{aba2008letter,wallison2008fair,wallison2008judgment}. For a review on the current state of the debate on FVA see \cite{laux2009crisis} and \cite{laux2010did} claiming that such an endogenous revaluation of inter-financial claims could lead to, so called, \textit{downward spirals}. A downward spiral occurs when a small shock decreasing the equity of a single bank triggers massive asset devaluations on the balance sheets of its counterparties, that are forced under CVA regulations to take into account its increased default probability (in agreement with \cite{merton1974pricing}). These devaluations, in turn, determine further equity losses that keep propagating to higher-order counterparties along the complex contractual interconnections of the financial system. In particular, notice that such a contagion mechanism does not require actual defaults to act as triggers for distress propagation, being able to generate substantial losses even in the absence of any bankruptcy. 
 The change in a counterparty's ability to repay may influence not only the mark-to-market value of a claim, but also the amount of collateral the counterparty may have to post: we do not tackle explicitly this problem in this paper, but the framework can be extended also to this case.

}

In the case of a debt security, $A_{ij}^b$, the expected value at time $t$ of an interbank debt claim of nominal value $A_{ij}^b$ will in general be given by
\begin{equation}
\label{eq:expected_value}
 \mathbb{V}_t(A_{ij}^b) = \underbrace{p_j^D(t) R A_{ij}^b}_{\text{payoff given default}} + (1 - p_j^D(t)) A_{ij}^b 
 \end{equation}

where $p_j^D(t)$ is the default probability of bank $j$ estimated at time $t$ and reflects the ability of the counterparty to honor its debts. 

This probability may reflect different \textit{sources of uncertainty} influencing the likelihood that the claim is actually paid (and in what amount). We envision here two main sources of uncertainty: the first source relates to the potential changes in the external assets $A_j^e$ of the counterparty, which can be modeled as a stochastic process: \cite{barucca2016network} provide a detailed account of this case and a formal relation with the standard structural Merton approach in a network context. The second source of uncertainty relates to the case in which the network of liabilities may be unknown \citep{haldane2009rethinking,caballero2013fire} or may evolve after the shock. In this work we do not model these cases explicitly, but we restrict our analysis to the fact that $\mathbb{V}_t(A_{ij}^b)$ is the expect value of a Bernoulli random variable as in Equation \ref{eq:expected_value}.

Several models are commonly used to estimate default probabilities.\footnote{E.g. approaches based on the Merton model \citep{merton1974pricing}. Other methods use implied default probabilities from the pricing of CDSs.} If we assume that the default probability of a bank is equal to its individual vulnerability\footnote{In \cite{battiston2016leveraging} it is pointed out that an alternative choice would be to maintain that $p_j^D(t) = f(h_j(t))$, where $f$ is a generic function of arbitrary form. The Acyclic DebtRank model, therefore, simply corresponds to the simplest linear choice for $f$.}, ie $p_j^D(t) = h_j(t)$, we obtain the following financial distress propagation process:
\begin{equation}
h_i(t+1) = \min \bigg\{1, h_i(t) + \sum_{j \in \mathcal{A}(t)} (1 - R)l_{ij}^b h_j(t) \bigg\} \nonumber
\end{equation}
where now, unlike the case of the DC model, the set of active nodes is defined as
\begin{equation}
 \mathcal{A}(t) = \{ j \:\:\: | \:\:\: h_j(t) > 0 \mbox{ and } h_j(t') = 0, \quad \forall t'<t \} \nonumber
 \end{equation}

This simple change in definition implies that banks can propagate distress as soon as they lose a fraction of their equity, even before their default (as required by CVA). Nevertheless, once a bank has propagated distress it will not be able to transmit further losses, despite still being able to receive them. This is equivalent to saying that in the graph process only first cycles are considered.

\subsubsection{Cyclic DebtRank (cDR)}
\label{subsubsec:si-cdr}

The Cyclic DebtRank model was proposed in \cite{bardoscia2015debtrank} as an extension of the Acyclic DebtRank model to allow distress to be propagated along all paths in the network, including all cycles. The extension also allows to overcome the mathematical intractabilities resulting from the  definition of the set of active nodes, $\mathcal{A}(t)$, in the Acyclic DebtRank model, thus allowing a full exploitation of the linearity underlying the model, with important consequences for the analysis of the stability of the system.\footnote{See  \citep{battiston2016financial}, \citep{bardoscia2015debtrank}, and \citep{bardoscia2016paths} for results concerning the relevance of the spectral properties of the interbank leverage matrix for the stability of the system} The process is defined as follows:
\vspace{-0.4cm}
$$ h_i(t+1) = \min \bigg\{1, h_i(t) + \sum_{j = 1}^n (1 - R)l_{ij}^b \big[ h_j(t) - h_j(t-1) \big] \bigg\} $$
where now the set of active nodes has been removed and summation runs over all counterparties.

\subsection{Analytical results on systemic losses}
\label{subsec:si-analytical-results}

\subsubsection{Ordering relations between EN, RV, DC, DR}
\label{subsubsec:si-relations-among-models}

We will now present a series of propositions in order to clarify the ordering relation among the different models in terms of individual and global vulnerabilities. This order relation holds regardless of the network structure.\footnote{The idea that the network structure alone does not determine the levels of contagion was explored in \cite{roukny2013default} by means of simulations.} 

The crucial step in this direction consists in re-expressing the Eisenberg-Noe and Rogers-Veraart models, normally formulated in terms of payment vectors, as recursive processes on individual vulnerabilities, $h_i(t)$. Proposition \ref{prop:EN-RV-leverage} partially provides such a variable transformation, highlighting at the same time the importance of leverage as a key quantity for quantifying distress propagation.

\begin{proposition}[EN/RV in terms of leverage]
\label{prop:EN-RV-leverage}
The Einseberg-Noe and Rogers-Veraart models can be re-expressed in terms of individual vulnerability as discrete-time processes on the interbank leverage network in the following way:

$$ h_i (t +  1) = \min \bigg\{ 1, h_i(t) + \sum_{j=1}^n l^b_{ij} \bigg( \frac{p_j(t - 1) - p_j(t)}{\bar{p}_j} \bigg) \bigg\} $$

where $p(t)$ is the payment vector at the $t$-th iteration of the Fictitious Default Algorithm.
\end{proposition}

\begin{proof}
To prove the claim we simply follow the definitions of individual vulnerability (Equation (\ref{eq:individual-vulnerability})) and net worth of a node (Equation (\ref{eq:en-equities-time-t})):

\begin{itemize} 
\item If bank $i$ has not defaulted up to time $t + 1$, we can write

\begin{eqnarray*}
h_i (t +  1) & = & \frac{E_i(0) - E_i(t + 1)}{E_i(0)}  =  \frac{E_i(0) - E_i(t)}{E_i(0)} + \frac{E_i(t) - E_i(t + 1)}{E_i(0)} =\\
& = & h_i (t) + \frac{1}{E_i(0)} \bigg( \sum_{j = 1}^n \Pi^T_{ij} p_j(t - 1) + \\
& & + A_i^e - p_i(t - 1) - \sum_{j = 1}^n \Pi^T_{ij} p_j(t) - A_i^e + p_i(t) \bigg).
\end{eqnarray*}

Notice that the equity at time $t$ is computed using the payment vector at time $t-1$, in agreement with the implementation of the Fictitious Default Algorithm.

Now, because bank $i$ has not defaulted by time $t+1$,   we know that its payment vector has not changed, hence $p_i(t) = p_i(t - 1)$. This implies

\begin{flalign*}
h_i (t +  1) & =  h_i (t) + \frac{1}{E_i(0)} \bigg( \sum_{j = 1}^n \Pi^T_{ij} (p_j(t - 1) - p_j(t)) \bigg) \\
& =  h_i (t) + \frac{1}{E_i(0)} \bigg( \sum_{j = 1}^n \frac{A^b_{ij}}{\bar{p}_j} (p_j(t - 1) - p_j(t)) \bigg) \\
& =  h_i (t) + \sum_{j = 1}^n \frac{A^b_{ij}}{E_i(0)} \frac{p_j(t- 1) - p_j(t)}{\bar{p}_j} \\
& =  h_i (t) + \sum_{j = 1}^n l^b_{ij} \frac{p_j(t - 1) - p_j(t)}{\bar{p}_j}.\nonumber
\end{flalign*}

\item If the bank $i$ defaulted at time $t$ or before, then $h_i(t) = 1$ and 
\begin{equation}
\min \bigg\{ 1, h_i (t) + \sum_{j = 1}^n l^b_{ij} \frac{p_j(t - 1) - p_j(t)}{\bar{p}_j} \bigg\} = 1
\nonumber
\end{equation}
so that our expression yields the correct result $h_i(t+1) = 1$, as required. 
\item If the default happens exactly at time $t + 1$ then it suffices to notice that the reduced inflow of payments on the interbank market has completely depleted the equity, that is

\begin{equation}
\sum_{j=1}^n \Pi^T_{ij} (p_j(t-1) - p_j(t)) \geq E_i(t) \nonumber
\end{equation}

this, in turn, implies that

\begin{equation}
h_i(t) + \frac{1}{E_i(0)} \sum_{j=1}^n \Pi^T_{ij} (p_j(t-1) - p_j(t)) \geq h_i(t) + \frac{E_i(t)}{E_i(0)} = 1 \nonumber
\end{equation}

so that, once again, by recognizing the leverage in the above expression, we have

\begin{equation}
\min \bigg\{ 1, h_i (t) + \sum_{j = 1}^n l^b_{ij} \frac{p_j(t - 1) - p_j(t)}{\bar{p}_j} \bigg\} = 1
\nonumber
\end{equation}
\end{itemize}

This completes our proof.
\end{proof}

Building on the results of Proposition \ref{prop:EN-RV-leverage}, we can now compare the different models, identifying partial order relations in terms of final equity losses generated.

Proposition \ref{prop:EN<RV} shows that the EN model yields lower vulnerabilities than the RV model. Its content is already well-known and is a simple consequence of the definition of the two models, but we re-present it here in terms of individual vulnerabilities for completeness.

In Proposition \ref{prop:RV<cDR} we prove that vulnerabilities generated by the cDR model are always greater than the ones generated by the RV model, while Propositions \ref{prop:not-DC<aDR} and \ref{prop:not-EN<aDR} show that no clear relationship exists between aDR and the remaining models can be derived, by exhibiting counterexamples.

It should be emphasized nevertheless that the counterexample provided in Proposition \ref{prop:not-EN<aDR} can be considered pathological, in the sense that initial balance sheet quantities have to be fine-tuned in such a way as to obtain particularly vulnerable institutions on the interbank market and carefully timed default waves (see comments at the end of the proof). In all empirical implementations, in fact, vulnerabilities induced by aDR are larger than the ones induced by EN.

\begin{proposition}[EN $\leq$ RV]
\label{prop:EN<RV}
At every time $t \geq 1$ the relation
\begin{equation}
h^{EN}_i(t) \leq h^{RV}_i(t) \nonumber
\end{equation}
is satisfied $\forall i \in V$, where $t$ is the internal time of the algorithms. In particular this implies that at convergence:
\begin{equation}
H^{EN}(\infty) \leq H^{RV}(\infty) \nonumber
\end{equation}
\end{proposition}

\begin{proof}
We will first of all prove by induction that at each iteration of the Fictitious Default Algorithm, the relation $p_i^{RV}(t) \leq p_i^{EN}(t)$ is satisfied. At the beginning of the algorithm, both models satisfy $p_i^{RV}(1) = p_i^{EN}(1) = \bar{p}_i$, thus constituting our inductive basis. Now, suppose that $p_i^{RV}(t) \leq p_i^{EN}(t)$ is true, we want to show that this implies $p_i^{RV}(t+1) \leq p_i^{EN}(t+1)$. This is immediate because by definition we have:
\begin{equation}
 p_i^{EN}(t+1) = \min \bigg\{\bar{p}_i, \sum_{j=1}^n \Pi^T_{ij} p_j^{EN}(t) + A_i^e \bigg\} \nonumber
 \end{equation}
and
\begin{equation}
p_i^{RV}(t+1) = \min \bigg\{\bar{p}_i, R \bigg( \sum_{j=1}^n \Pi^T_{ij} p_j^{RV}(t) + A_i^e \bigg) \bigg\} \nonumber
\end{equation}

Given that $ \beta \in [0, 1]$, it follows that $p_i^{RV}(t+1) \leq p_i^{EN}(t+1)$, thus proving our claim. We now use this to prove that the inverse relationship holds on individual vulnerabilities, for every $t \geq 1$.

\begin{itemize}

\item If bank $i$ has defaulted under the EN algorithm at time $t$ or before, we have that $p_i^{EN}(t) \leq \bar{p}_i$. By what we have just proved on payment vectors, we know that also $p_i^{RV}(t) \leq \bar{p}_i$ and bank $i$ has defaulted under RV as well. Therefore $h_i^{EN}(t+1) = h_i^{RV}(t+1) = 1$, because vulnerabilities are non-decreasing in $t$.

\item If bank $i$ has not defaulted at time $t$ or before, we can write $h_i(t)$ in Proposition \ref{prop:EN-RV-leverage} explicitly, by solving backwards in time, thus obtaining:

\begin{equation}
h_i^{EN} (t +  1) = \min \bigg\{ 1, h_i^{EN}(1) + \sum_{j=1}^n l^b_{ij} \bigg( \frac{\bar{p}_j - p_j^{EN}(t)}{\bar{p}_j} \bigg) \bigg\}, \nonumber
\end{equation}

\begin{equation}
h_i^{RV} (t +  1) = \min \bigg\{1, h_i^{RV}(1) + \sum_{j=1}^n l^b_{ij} \bigg( \frac{\bar{p}_j - p_j^{RV}(t)}{\bar{p}_j} \bigg)\bigg\}. \nonumber
\end{equation}

Clearly, since we have established that $p_i^{RV}(t) \leq p_i^{EN}(t)$, for every $t$, and we know that $h_i^{RV}(1) = h_i^{EN}(1)$, we must have $h_i^{EN}(t+1) \leq h_i^{RV}(t+1)$.
\end{itemize}
\end{proof}

A straightforward corollary, as already showed in \cite{rogers2013failure} within their framework, is that if $\beta = 1$, then $h_i^{EN} (t) = h_i^{RV} (t), \forall t$, which implies that $H^{EN}(\infty)= H^{RV}(\infty)$ if $\beta = 1$.

\begin{proposition}[RV $\leq$ cDR]
\label{prop:RV<cDR}
At every time $t \geq 1$ the relation
\begin{equation}
h^{RV}_i(t) \leq h^{cDR}_i(t) \nonumber
\end{equation}

is satisfied $\forall i \in V$, where $t$ is the internal time of the algorithms, which implies that at convergence:
\begin{equation}
H^{RV}(\infty) \leq H^{cDR}(\infty) \nonumber
\end{equation}
\end{proposition}

\begin{proof}
We prove the claim by induction. For $t = 1$, we have that
\begin{equation}
h^{RV}_i(1) = h^{cDR}_i(1) \nonumber
\end{equation}

since first-round losses are identical for all models. Suppose the inequality holds at time $t$, we  show that this holds true also at time $t + 1$. In fact:

\begin{itemize}

\item if $h^{cDR}_i(t) = 1 \Longrightarrow h^{RV}_i(t + 1) \leq 1 = h^{cDR}_i(t + 1)$, because vulnerabilities are bounded above by $1$;
\item if  $h^{cDR}_i(t) < 1$ then the general formulas read
\begin{equation}
\begin{array}{rcl}
 h^{cDR}_i(t + 1) & =& \min \bigg\{ 1, h^{cDR}_i(t) + \sum_{j=1}^n l^b_{ij} \big( h^{cDR}_j(t) - h^{cDR}_j(t-1) \big) \bigg\},  \\
 h^{RV}_i (t +  1) & =& \min \bigg\{ 1, h^{RV}_i(t) + \sum_{j=1}^n l^b_{ij} \bigg( \frac{p_j(t - 1) - p_j(t)}{\bar{p}_j} \bigg) \bigg\}.
 \end{array}
\nonumber
\end{equation}
But since $h^{RV}_i(t) \leq h^{cDR}_i(t) < 1$ we can solve the telescopic sum in the second term of the minimum operator, obtaining:\footnote{It is assumed that $h^{cDR}_i(0) = h^{RV}_i(0) = 0$, ie. no bank is in default prior to the first round shock. This amounts to exclude from the network all banks reporting a negative book-value of equity.
}
\begin{equation}
\begin{array}{rcl}
h^{cDR}_i(t + 1) &=& \min \bigg\{ 1, h^{cDR}_i(1) + \sum_{j=1}^n l^b_{ij} h^{cDR}_j(t) \bigg\},  \\
 h^{RV}_i (t +  1) &=& \min \bigg\{ 1, h^{RV}_i(1) + \sum_{j=1}^n l^b_{ij} \bigg( \frac{\bar{p}_j - p_j(t)}{\bar{p}_j} \bigg) \bigg\}.
 \end{array}
\nonumber
\end{equation}
We know by the inductive step that $ h^{RV}_i(1) = h^{cDR}_i(1) $, so we will focus on the terms in the summation. Now, notice that if counterparty $j$ has defaulted under the cDR algorithm by time $t$, then:
\begin{equation}
h^{cDR}_j (t) = 1 \geq \frac{\bar{p}_j - p_j(t)}{\bar{p}_j}, \nonumber
\end{equation}
which follows from the fact that $p_j(t) \in [0, \bar{p}_j]$. On the other hand, if counterparty $j$ has not defaulted under the cDR algorithm by time $t$, then
\begin{equation}
 h^{RV}_j (t) \leq h^{cDR}_j (t) < 1 \Longrightarrow p_j(t) = \bar{p}_j, \nonumber
 \end{equation}
where the first inequality is true for the inductive hypothesis and the implication is a simple consequence of the Fictitious Default Algorithm: if bank $j$ is not in default, then by priority of claim it must be able to pay its counterparties in full. Once again, we obtain: 
\begin{equation}
h^{cDR}_j (t) \geq \frac{\bar{p}_j - p_j(t)}{\bar{p}_j} = 0 \nonumber
 \end{equation}
Putting these two inequalities together and recalling that $ h^{RV}_i(1) = h^{cDR}_i(1) $ is true, we finally obtain:
\begin{equation}
 h^{RV}_i(t + 1) \leq h^{cDR}_i(t + 1). \nonumber
\end{equation}
\end{itemize} \end{proof}

\begin{proposition}[DC $\nleq$ aDR]
\label{prop:not-DC<aDR}
In general it does not hold that:

$$ H^{DC}(\infty) \leq H^{aDR}(\infty) $$
\end{proposition}

\begin{proof}
We will prove this by exhibiting a counterexample. Consider the following financial system:

$$ E(0) = \begin{bmatrix}
       5 \\
       15 \\
       25
     \end{bmatrix} \:\:\:
	A^{(e)} = \begin{bmatrix}
       100 \\
       100 \\
       100
     \end{bmatrix}\:\:\:
     A^{(b)} = \begin{bmatrix}
       0 & 0 & 20 \\
       20 & 0 & 0 \\
       0 & 15 & 0
     \end{bmatrix}$$

\begin{center}
\begin{tikzpicture}[->,>=stealth',shorten >=1pt,auto,node distance=3cm,
                    thick,main node/.style={circle,draw,font=\sffamily\Large}]

  \node[main node] (1) {1};
  \node[main node] (2) [below left of=1] {2};
  \node[main node] (3) [below right of=1] {3};

  \path[every node/.style={font=\sffamily\small}]
    (1) edge [bend right] node[left] {20} (2)
    (2) edge [bend right] node[below] {15} (3)
    (3) edge [bend right] node[right] {20} (1);
\end{tikzpicture}
\end{center}

If we shock all external assets by $10$\%, we obtain

$$ h^{aDR}(1) = \begin{bmatrix}
       1 \\
       2/3 \\
       2/5
     \end{bmatrix} \:\:\:
     h^{aDR}(2) = \begin{bmatrix}
       1 \\
       1 \\
       4/5
     \end{bmatrix}  \: \mbox{ and }$$

$$ h^{DC}(1) = \begin{bmatrix}
       1 \\
       2/3 \\
       2/5
     \end{bmatrix} \:\:\:
     h^{DC}(2) = \begin{bmatrix}
       1 \\
       1 \\
       2/5
     \end{bmatrix}\:\:\:
     h^{DC}(3) = \begin{bmatrix}
       1 \\
       1 \\
       1
     \end{bmatrix}$$
     
So that $H^{DC}(\infty) > H^{aDR}(\infty)$. 

The counterexample works by constructing the network in such a way that two waves of distress (one moderate and one more substantial) affect the system. Bank $2$ is sent into default by bank $1$, but it is unable to propagate this distress onto $3$ under aDR, because it has already propagated the lesser distress coming from the first wave. Under $DC$, on the other hand, bank $2$ can now propagate its distress, which more than makes up for the ``lack of propagation'' of the first wave.

It is worthwhile to point out that this example is not pathological in any sense: it is routinely observed in our empirical results that, past a given external shock threshold (high enough to send some banks into default at the first round), DC can in fact produce much higher vulnerabilities than aDR. 
\end{proof}

\begin{proposition}[EN $\nleq$ aDR]
\label{prop:not-EN<aDR}
In general it does not hold that:
\begin{equation}
H^{EN}(\infty) \leq H^{aDR}(\infty) \nonumber
 \end{equation}
\end{proposition}

\begin{proof}
We will prove this by exhibiting a counterexample. Consider the following financial network:

$$ E(0) = \begin{bmatrix}
       15 \\
       35 \\
       35
     \end{bmatrix} \:\:\:
	A^{(e)} = \begin{bmatrix}
       100 \\
       5 \\
       20
     \end{bmatrix}\:\:\:
     A^{(b)} = \begin{bmatrix}
       0 & 0 & 0 \\
       50 & 0 & 0 \\
       0 & 20 & 0
     \end{bmatrix}$$

\begin{center}
\begin{tikzpicture}[->,>=stealth',shorten >=1pt,auto,node distance=3cm,
                    thick,main node/.style={circle,draw,font=\sffamily\Large}]

  \node[main node] (1) {1};
  \node[main node] (2) [right of=1] {2};
  \node[main node] (3) [right of=2] {3};

  \path[every node/.style={font=\sffamily\small}]
    (1) edge [right] node[below] {50} (2)
    (2) edge [right] node[below] {20} (3);
\end{tikzpicture}
\end{center}

If we shock all external assets by $100$\%, we obtain

$$ h^{aDR}(1) = \begin{bmatrix}
       1 \\
       1/7 \\
       4/7
     \end{bmatrix} \:\:\:
     h^{aDR}(2) = \begin{bmatrix}
       1 \\
       1 \\
       32/49
     \end{bmatrix}, \: \mbox{ and}$$

$$ h^{EN}(1) = \begin{bmatrix}
       1 \\
       1/7 \\
       4/7
     \end{bmatrix} \:\:\:
     h^{EN}(2) = \begin{bmatrix}
       1 \\
       1 \\
       4/7
     \end{bmatrix}\:\:\:
     h^{EN}(3) = \begin{bmatrix}
       1 \\
       1 \\
       1
     \end{bmatrix}.$$
     
So that $H^{EN}(\infty) > H^{aDR}(\infty)$. 

This counterexample exploits the same features analyzed in Proposition \ref{prop:not-DC<aDR}. Clearly in order for the same principle to work in the case of the EN algorithm, the financial system has to be particularly vulnerable to interbank contagion, so that this counterexample cannot be considered typical of empirical data in any sense and it has been built \textit{ad hoc}. In actual implementations, we see that aDR always leads to higher vulnerabilities. This counterexample, in fact, is significant only insofar as it increases our understanding of the limitations of the aDR dynamics w.r.t. the cDR dynamics. In particular, in Section \ref{subsubsec:si-cdr}, while introducing the cDR model, we briefly pointed out the importance of this definition of active nodes.\end{proof}

\subsubsection{Conservation of losses in EN}
\label{subsec:si-cons-loss-en}

We will now derive analytical results for the Eisenberg-Noe model that are key to our analysis. In particular in Proposition \ref{prop:EN-final-H} we derive an exact expression for the final vulnerability of the system at convergence, $H^{EN}(\infty)$, that will be subsequently used in Proposition \ref{prop:EN-second-round-upper-bound} to obtain an upper bound for the second round losses in equity.

\begin{proposition}[Analytical expression for $H^{EN}(\infty)$]
\label{prop:EN-final-H}
Let $s_i$ be the shock on the external assets of the $i$-th bank, then

$$ H^{EN}(\infty) = \frac{1}{\sum_{i=1}^n E_i(0)} \sum_{i=1}^n \bigg( A^{(e)}_i s_i - (1 - \beta_i) (\bar{p}_i - p_i(\infty)) \bigg)$$ 

where $p(\infty)$ is the clearing payment vector at convergence and $\beta_i := L_i^{(b)} / (L_i^{(b)} + L_i^{(e)}) \in [0, 1]$ is the ratio\footnote{This coefficient is obviously of great importance, as already highlighted by \cite{glasserman2015likely}, and is known as \textit{financial connectivity}} between interbank liabilities and total liabilities of bank $i$.
\end{proposition}

\begin{proof}
At convergence the clearing payment vector satisfies Equation \ref{eq:en-clearing-payment-vector} in Section \ref{subsubsec:si-en}, i.e.

$$p_i(\infty) = \min \bigg\{\sum_{j = 1}^n \Pi^T_{ij} p_j(\infty) + A_i^e (1 - s_i), \bar{p}_i \bigg\} $$

so that the final equity of bank $i$ is

$$E_i(\infty) = \sum_{j = 1}^n \Pi^T_{ij} p_j(\infty) + A_i^{(e)} (1 - s_i) - p_i(\infty)$$

We also notice that at $t = 0$ we have:\footnote{Assuming no bank is in default prior to the first round shock}
\begin{eqnarray*}
E_{tot}(0) & = & \sum_{i=1}^n E_i(0) = \sum_{i=1}^n \bigg( \sum_{j = 1}^n \Pi^T_{ij} \bar{p}_j + A_i^{(e)} - \bar{p}_i \bigg) =\\
& = &  \sum_{j = 1}^n \sum_{i=1}^n \Pi^T_{ij} \bar{p}_j + \sum_{i=1}^n A_i^{(e)} - \sum_{i=1}^n \bar{p}_i  =\\
& = & \sum_{j = 1}^n \beta_j \bar{p}_j + \sum_{i=1}^n A_i^{(e)} - \sum_{i=1}^n \bar{p}_i  =\\
& = & \sum_{j = 1}^n (\beta_j - 1) \bar{p}_j + \sum_{i=1}^n A_i^{(e)}
\end{eqnarray*}

Putting the two equations together, we obtain

\begin{eqnarray*}
H^{EN}(\infty) & = & \frac{E_{tot}(0) - E_{tot}(\infty)}{E_{tot}(0)} =\\
& = & \frac{1}{{\sum_{i=1}^n E_{tot}(0)}} \sum_{i=1}^n \bigg( (\beta_i - 1) \bar{p}_i + A_i^{(e)} - (\beta_i - 1) p_i(\infty) - A_i^{(e)} (1 - s_i) \bigg) =\\
& = & \frac{1}{{\sum_{i=1}^n E_{tot}(0)}} \sum_{i=1}^n \bigg( A_i^{(e)} s_i - (1 - \beta_i) (\bar{p}_i - p_i(\infty))\bigg).
\end{eqnarray*}

Therefore the total cumulative relative equity loss of the system equals the total loss in assets at the first round minus the shortfall in payments to external debtholders, that are forced to endure a certain amount of writing-off of debt.
\end{proof}

\begin{proposition}[Second round losses in EN]
\label{prop:EN-second-round}
Let $s_i$ be the shock on the external assets of the $i$-th bank, then the exact expression for the cumulative relative equity loss incurred as second round effect is given by

$$ H^{EN}(\infty) - H^{EN}(1)  = \frac{1}{\sum_{i=1}^n E_i(0)} \bigg[ \sum_{i \in \mathcal{D}(1)} \bigg( A^{(e)}_i s_i - E_i(0) \bigg) -  \sum_{i = 1}^n (1 - \beta_i) (\bar{p}_i - p_i(\infty))\bigg]  $$ 

where $\mathcal{D}(1) = \{i \in V | s_i > 1/l_i^e\}$ is the set of defaulted banks at the first round. 

The result can be interpreted as saying that total equity loss during the second round is exactly equal to the loss in assets at the first round in excess of equity minus the shortfall in payments to external debtholders, which represents that fraction of the initial shock that is ultimately absorbed by entities outside of the network.

\end{proposition}

\begin{proof}

From Proposition (\ref{prop:EN-final-H}) we know

$$ H^{EN}(\infty) = \frac{1}{\sum_{i=1}^n E_i(0)} \sum_{i=1}^n \bigg( A^e_i s_i - (1 - \beta_i) (\bar{p}_i - p_i(\infty)) \bigg)$$ 

A direct computation of $H^{EN}(1)$ yields
\begin{equation}
 H^{EN}(1) = \frac{1}{\sum_{i=1}^n E_i(0)} \sum_{i=1}^n \big( E_i(0) \min \left\{1, l_i^e s_i \right\} \big), \nonumber
 \end{equation}
where $l_i^e := A_i^e / E_i(0)$ is the external leverage of bank $i$.
Therefore:
\begin{eqnarray*}
H^{EN}(\infty) - H^{EN}(1) & = & \frac{1}{\sum_{i=1}^n E_i(0)} \bigg[ \sum_{i=1}^n \bigg( A^e_i s_i - (1 - \beta_i) (\bar{p}_i - p_i(\infty)) \bigg) + \\
& & \qquad - \sum_{i=1}^n E_i(0) \min \{1, l_i^e s_i \} \bigg] =  \\
& = &\frac{1}{\sum_{i=1}^n E_i(0)} \bigg[ \sum_{i=1}^n \bigg( A^e_i s_i - (1 - \beta_i) (\bar{p}_i - p_i(\infty)) \bigg) + \\
& & \qquad -\sum_{i \in \mathcal{D}(1)} E_i(0) - \sum_{i \in \mathcal{D}(1)^c} A_i^{(e)} s_i \bigg] =\\
& = & \frac{1}{\sum_{i=1}^n E_i(0)} \bigg[ \sum_{i \in \mathcal{D}(1)} \bigg( A_i^{e} s_i - E_i(0) \bigg) -  \sum_{i=1}^n (1 - \beta_i) (\bar{p}_i - p_i(\infty)) \bigg],
\end{eqnarray*}

which proves our claim. 
\end{proof}

Notice that the expression proved in Proposition \ref{prop:EN-second-round} will in general depend on the clearing payment vector $p(\infty)$, since we cannot know in advance how much of the external debt will be written off. We are therefore interested in finding an upper bound of this quantity that is independent of $p(\infty)$ at convergence.

\begin{proposition}[Upper bound on second round losses in EN]
\label{prop:EN-second-round-upper-bound}
Under the hypotheses of Proposition \ref{prop:EN-second-round} we have an upper bound on the losses at the second round given by:
\begin{equation}
 H^{EN}(\infty) - H^{EN}(1)  \leq \frac{1}{\sum_{i=1}^n E_i(0)}  \sum_{i \in \mathcal{D}(1)} \beta_i \bigg( A^{e}_i s_i - E_i(0) \bigg). \nonumber
 \end{equation}
\end{proposition}

\begin{proof}
Without knowing anything about $p(\infty)$, we can nevertheless say that:
\begin{equation}
\bar{p}_i - p_i(\infty) \geq
\begin{cases}
A_i^{e} s_i - E_i(0)  & \:\:\: \mbox{  if $i \in \mathcal{D}(1)$ } \\
0 & \:\:\: \mbox{ otherwise } 
\end{cases}, \nonumber
\end{equation}

in fact those banks that defaulted at the first round will have at least a shortfall in payments equal to the amount of lost assets in excess of equity, while we conservatively assume that all other banks will not default after the first round. This conservative estimates leads us to an upper bound for the second round effects given by:
\begin{eqnarray*}
H^{\text{EN}}(\infty) - H^{\text{EN}}(1) & = & \frac{1}{\sum_{i=1}^n E_i(0)} \bigg[ \sum_{i \in \mathcal{D}(1)} \bigg( A_i^{e} s_i - E_i(0) \bigg) -  \sum_{i=1}^n (1 - \beta_i) (\bar{p}_i - p_i(\infty)) \bigg] = \\
& \leq & \frac{1}{\sum_{i=1}^n E_i(0)} \bigg[ \sum_{i  \in \mathcal{D}(1)} \bigg( A_i^{e} s_i - E_i(0) \bigg) + \\
& & \qquad -  \sum_{i \in \mathcal{D}(1)} (1 - \beta_i) (A_i^{e} s_i - E_i(0)) \bigg] = \\
& = & \frac{1}{\sum_{i=1}^n E_i(0)} \sum_{i \in \mathcal{D}(1)} \beta_i \bigg( A_i^{e} s_i - E_i(0) \bigg).
\end{eqnarray*}

The intuition behind this upper bound is that the second round losses absorbed by the system will be equal to at most a fraction of the losses in assets in excess of equity, the factor of proportionality being exactly the financial connectivity coefficients of the banks hit at the first round. 
\end{proof}


Further calculations on the bound provided in Proposition \ref{prop:EN-second-round-upper-bound} lead to the rewrite the upper bound as follows:
\begin{eqnarray*}
H^{\text{EN}}(\infty) - H^{\text{EN}}(1) & \leq &\frac{1}{\sum_{i=1}^n E_i(0)} \sum_{i \in \mathcal{D}(1)} \beta_i \bigg( A_i^{e} s_i - E_i(0) \bigg) =  \\
& =& \sum_{i \in \mathcal{D}(1)} \beta_i \bigg( l_i^{e} \frac{E_i(0)}{\sum_{j=1}^n E_i(0)} s_i - \frac{E_i(0)}{\sum_{i=1}^n E_i(0)} \bigg) = \\
&=& \sum_{i \in \mathcal{D}(1)}\beta_i w_i  \left(l_i^e s_i - 1\right)
\end{eqnarray*}

where $l_i s_i - 1$ is the the \textit{loss in excess of equity} and it's always larger than zero by the way $\mathcal{D}(1)$ is defined, and $w_i = E_i(0) / E_{tot}(0)$ is the weight of bank $i$ in terms of total network equity. This formulation allows to interpret the second round losses as the weighted average of the individual excess losses at the first round on the subset of banks defaulted at the first round.

%
%
%
%
%

\subsection{Illustrative examples on EN}
\label{subsec:si-illustrative-examples}

\subsubsection{Mutualization of losses}

We now present some examples in order to clarify the process of losses mutualization in the network under the EN model. Consider a directed wheel graph on $n$ vertices, with $n \geq 2$, shown in Figure \ref{fig:wheel} for the case $n=4$. Order the vertices $\{1, \ldots, n\}$. Suppose there are two classes of banks with different balance sheets. The first node, occupying the center of the wheel, is a fragile bank, with balance sheet:
\begin{equation}
\begin{array}{c}
A_1^e = 75(n-1), \: \: A_1^b = 0 \\
E_1(0) = 5(n-1), \: \: L_1^b = 10(n-1), \: \: L_1^e = 60(n-1)
\end{array} \nonumber
\end{equation}
while the remaining $n-1$ banks are more capitalized:
\begin{equation}
\begin{array}{c}
A_i^e = 50(n-1), \: \: A_i^b = 15(n-1) \\
E_i(0) = 10(n-1), \: \: L_i^b = 5(n-1), L_i^e = 50(n-1), \: \: \: \mbox{for } i = 2, \ldots, n
\end{array} \nonumber
\end{equation}
The balance sheet quantities are consistent for every value of $n \geq 2$ and are linearly increasing with the size of the network. The figure below shows such a wheel graph for $n=4$. The bank shown in red is the fragile bank at the center of the wheel.

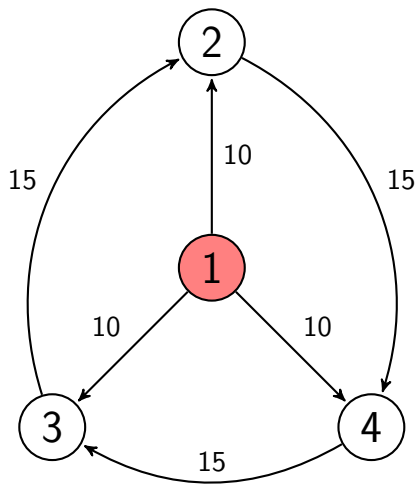
\begin{figure}[H]
\begin{center}
\begin{tikzpicture}[->,>=stealth',shorten >=1pt,auto,node distance=3cm,
                    thick,main node/.style={circle,draw,font=\sffamily\Large}]

  \node[main node] (2) {2};
  \node[main node, fill=red!50] (1) [below of=2] {1};
  \node[main node] (3) [below left of=1] {3};
  \node[main node] (4) [below right of=1] {4};

  \path[every node/.style={font=\sffamily\small}]
    (1) edge [right] node[right] {10} (2)
    (1) edge [right] node[above left] {10} (3)
    (1) edge [right] node[above right] {10} (4)
    (3) edge [bend left=40] node[above left] {15} (2)
    (2) edge [bend left=40] node[above right] {15} (4)
    (4) edge [bend left=30] node[above] {15} (3);
\end{tikzpicture}
\caption{Wheel graph for $n=4$.}
\label{fig:wheel}
\end{center}
\end{figure}

Assume an initial shock of $s = 10\%$ impacts the external assets of the first bank, making it default. We can readily compute the individual vulnerability of the remaining banks

$$ h^{EN}_i(\infty) = \frac{(A_1^e s - E_1(0) )\beta_1}{n-1}\cdot \frac{1}{E_i(0)} = \frac{3.75 \%}{n-1} , \: \: \: \mbox{for } i=2, \ldots, n$$

Notice that the $n-1$ factor at the denominator implies that all losses in excess of bank $1$ equity are fully mutualized among the counterparties. Computing the final global vulnerability yields to:
\begin{align}
 H^{EN}(\infty) & =  (n - 1) \frac{10(n-1)}{10 (n-1)^2  + 5(n-1)} \cdot \frac{3.57\%}{(n-1)}  + \nonumber \\ & + \frac{5(n-1)}{10(n-1)^2 + 5(n-1)} =  \frac{1.075}{2(n-1) + 1}. \nonumber
 \end{align}
As the number of counterparties increases the default of the more fragile institution induces less and less distress in the network, since losses are mutualized among a greater number of counterparties.

\subsubsection{Invariance with respect to network topology}

From Proposition \ref{prop:EN-second-round-upper-bound}, we know that the second round global vulnerability under the EN model admits the following bound, \emph{independently of the underlying network topology}
\begin{equation}
H^{EN}(\infty) - H^{EN}(1)  \leq \frac{1}{\sum_{i=1}^n E_i(0)}  \sum_{i \in \mathcal{D}(1)} \beta_i \bigg( A^{(e)}_i s_i - E_i(0) \bigg)\nonumber 
\end{equation}

We will now apply this bound to three simple networks with different topologies, as in Figure \ref{fig:topologies}, and show how the second round under EN is in fact \textit{independent} of the network structure.

\begin{figure}[H]
\centering
\begin{subfigure}[t]{\textwidth}
\centering
\begin{tikzpicture}[->,>=stealth',shorten >=1pt,auto,node distance=3cm,
                    thick,main node/.style={circle,draw,font=\sffamily\Large}]

  \node[main node, fill=red!50] (1) {1};
  \node[main node] (2) [right of=1] {2};
  \node[main node] (3) [right of=2] {3};
  \node[main node] (4) [right of=3] {4};

  \path[every node/.style={font=\sffamily\small}]
    (1) edge [right] node[below] {15} (2)
    (2) edge [right] node[below] {15} (3)
    (3) edge [right] node[below] {15} (4);
\end{tikzpicture}
\caption{Chain graph}
\label{fig:chain}
\end{subfigure}
\begin{subfigure}[t]{0.4\textwidth}
\centering
\begin{tikzpicture}[->,>=stealth',shorten >=1pt,auto,node distance=3cm,
                    thick,main node/.style={circle,draw,font=\sffamily\Large}]

  \node[main node] (2) {2};
  \node[main node, fill=red!50] (1) [below of=2] {1};
  \node[main node] (3) [below left of=1] {3};
  \node[main node] (4) [below right of=1] {4};

  \path[every node/.style={font=\sffamily\small}]
    (1) edge [right] node[right] {5} (2)
    (1) edge [right] node[above left] {5} (3)
    (1) edge [right] node[above right] {5} (4);
\end{tikzpicture}
\caption{Star graph}
\label{fig:star}
\end{subfigure}
\begin{subfigure}[t]{0.4\textwidth}
\centering
\begin{tikzpicture}[->,>=stealth',shorten >=1pt,auto,node distance=3cm,
                    thick,main node/.style={circle,draw,font=\sffamily\Large}]

  \node[main node, fill=red!50] (1) {1};
  \node[main node] (2) [left of=1] {2};
  \node[main node] (3) [below of=2] {3};
  \node[main node] (4) [right of=3] {4};

  \path[every node/.style={font=\sffamily\small}]
    (1) edge [bend right] node[above] {15} (2)
    (2) edge [bend right] node[left] {15} (3)
    (3) edge [bend right] node[below] {15} (4)
    (4) edge [bend right] node[right] {15} (1);
\end{tikzpicture}
\caption{Directed cyclic graph}
\label{fig:cycle}
\end{subfigure}
\caption{Examples of different topologies on simple 4-vertices graphs}
\label{fig:topologies}
\end{figure}
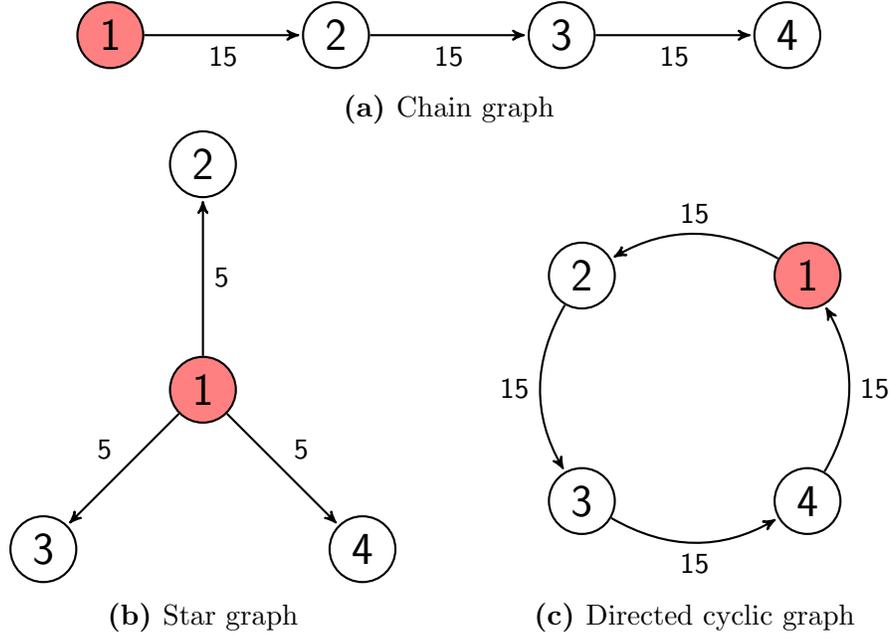

In each of the graphs shown, the first node (colored in red) corresponds to a fragile bank with balance sheet
\begin{equation}
 A_1^e= 80, \: \: L_1^e = 60, \: \: L^b_1 = 15, \: \: E_1(0) = 5. \nonumber \end{equation}

The remaining three banks have equities $E_i(0) = 10, i = 2, 3, 4$, while the interbank linkages are shown on the graphs themselves. Assume that an initial shock $s = 10\%$ impacts only the external assets of bank $1$, leading to its default. Before considering the topology of the network and only by looking at the balance sheet quantities provided, we know that:
\begin{eqnarray*}
H^{EN}(\infty) - H^{EN}(1) & \leq & \frac{1}{\sum_{i=1}^n E_i(0)}  \sum_{i \in \mathcal{D}(1)} \beta_i \bigg( A^{(e)}_i s_i - E_i(0) \bigg) \\
& = & \frac{1}{35} \cdot \frac{15}{60} \cdot 4 = 1.71\% .
\end{eqnarray*}

In other words, the second round losses cannot be higher than $1.71\%$ of the total initial equity. Indeed, a direct computation of the individual vulnerabilities induced by the initial shock in the cases of Figures \ref{fig:chain}, \ref{fig:star}, and \ref{fig:cycle} leads to the following results
$$ \mbox{\textbf{a)} } h^{EN}(\infty) = \begin{bmatrix}
       1 \\
       0.06 \\
       0 \\
       0 \\
     \end{bmatrix}, \:\:\:
     \mbox{\textbf{b)} }
     h^{EN}(\infty) = \begin{bmatrix}
       1 \\
       0.02 \\
       0.02 \\
       0.02
     \end{bmatrix},\:\:\:
     \mbox{\textbf{c)} }
     h^{EN}(\infty) = \begin{bmatrix}
       1 \\
       0.06 \\
       0 \\
       0
     \end{bmatrix}.$$
In all three cases, despite the differences in the vulnerabilities of individual banks and independently from the topology, the final result is:
\begin{equation}
H^{EN}(\infty) - H^{EN}(1) = 1.71 \% \: \: \: \mbox{ and } \: \: \: H^{EN}(\infty) = 16 \% \nonumber
\end{equation}
that is, the bound is satisfied strictly in the case of all three topologies. 

The reason why the bound is achieved exactly is that, in the particular cases of these examples, only \textit{one wave} of defaults is triggered, namely the default of bank $1$ at time $t=1$. If further defaults are induced, the second round is even smaller due to the fact that at each successive default a fraction of losses is absorbed by the external debtholders of the defaulting institution, thus reducing the conservative flow in and absorption by the network.

\begin{figure}[H]
\centering
\includegraphics[width=\textwidth]{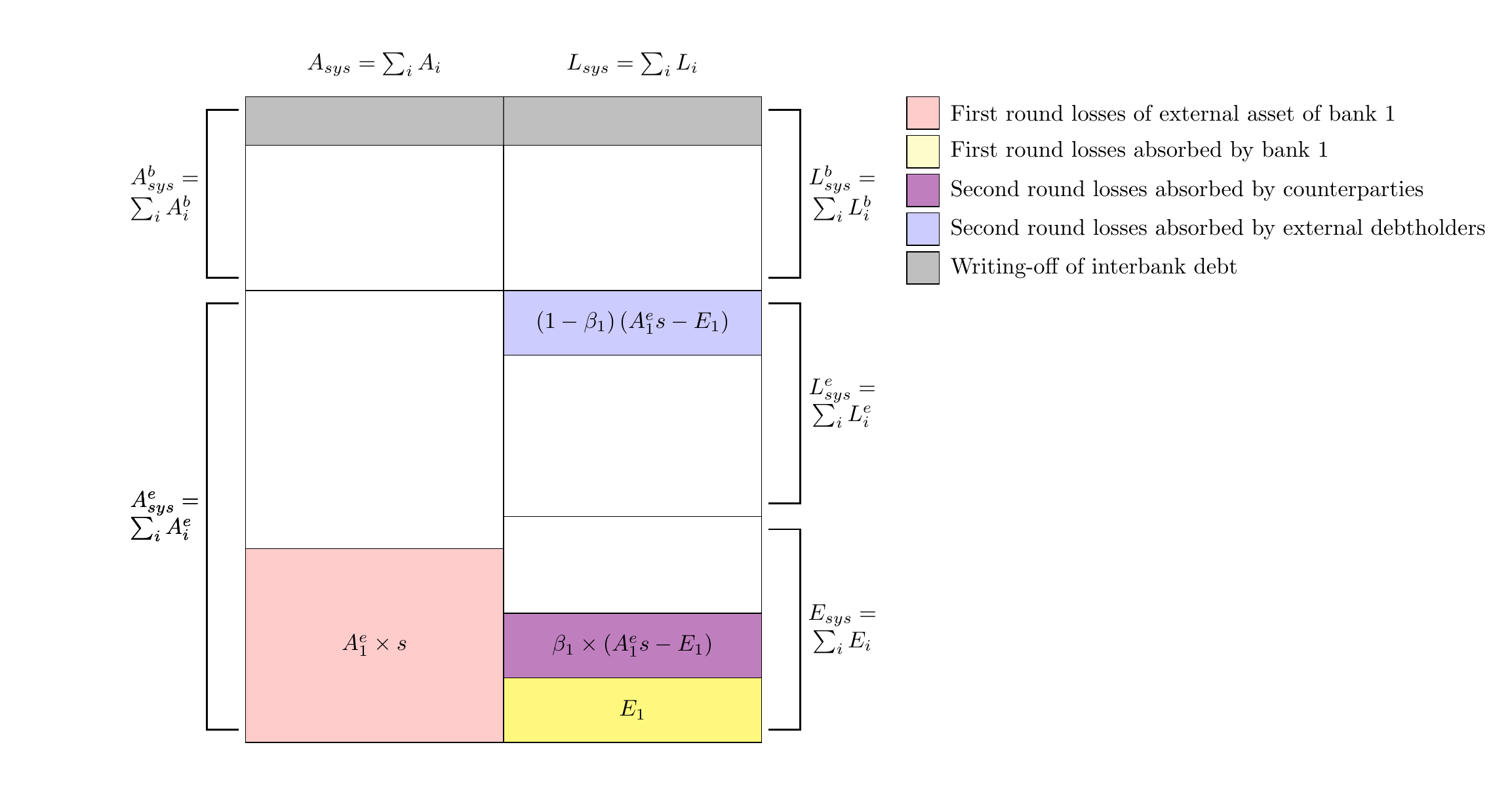}
\caption{Aggregated balance sheet for networks in Figures \ref{fig:chain}, \ref{fig:star}, \ref{fig:cycle} and corresponding dynamics for the EN model. Notice that the dynamics is completely network-independent}
\label{fig:agg-balance-sheet}
\end{figure}

It's worth emphasizing that the final global vulnerability value is the same as if we aggregated all balance sheets, as in Figure \ref{fig:agg-balance-sheet}. 

The initial losses induced by the shock in external assets of bank $1$, are partially shoulders by bank $1$ equity. Then a fraction $\beta_1$ of losses in excess is passed to its counterparties, while a fraction $1 - \beta_1$ is passed onto the external debtholders, thus leaving the network. The losses affecting counterparties then cascade through the network and, if they cause further defaults, are proportionally diminished by the fraction absorbed by external debtholders.

To see how the underlying topology could affect and truly amplify the propagation of losses, let's analyze the case of the aDR model, in which the results are clearly topology-dependent. From simple computations we obtain:

$$ \mbox{\textbf{a)} } h^{aDR}(\infty) = \begin{bmatrix}
       1 \\
       0.75 \\
       0.56 \\
       0.42 \\
     \end{bmatrix}, \:\:\:
     \mbox{\textbf{b)} }
     h^{aDR}(\infty) = \begin{bmatrix}
       1 \\
       0.75 \\
       0.75 \\
       0.75
     \end{bmatrix},\:\:\:
     \mbox{\textbf{c)} }
     h^{aDR}(\infty) = \begin{bmatrix}
       1 \\
       0.75 \\
       0.56 \\
       0.42
     \end{bmatrix}$$

so that

$$ \mbox{\textbf{a)} } H^{aDR}(\infty) = 0.64,\:\:\:
     \mbox{\textbf{b)} }
     H^{aDR}(\infty) = 0.79,\:\:\:
     \mbox{\textbf{c)} }
     H^{aDR}(\infty) = 0.64$$

In the case of the aDR model the advantages and disadvantages of the different topologies emerge clearly. Naturally enough aDR is a first-cycles process, so that the chain in Figure \ref{fig:chain} and the cycle in Figure \ref{fig:cycle} produce exactly the same final vulnerabitlity. But the star graph in Figure \ref{fig:star} leads to higher vulnerabilities, because the fragile node is able to trasmit its distress to the highest possible number of counterparties. Indeed the systemic risk posed by a star network configuration of this kind is well known and widely debated in the literature on CCPs (Central Clearing Counterparties).

\subsection{Data collection and processing}
\label{subsec:si-data-collection}

Our dataset comes from the Bureau Van Dijk Bankscope database\footnote{URL: bankscope.bvdinfo.com}. We focus on a representative subset of the 285 EU listed banks, for which we obtained individual balance sheet data\footnote{In details: we downloaded the following fields from the Universal Banking Model (UBM) of Bankscope. 1) ``Total Equity'', 2) ``Total assets'', 3) ``Loans and Advances to Banks'', 4) ``Deposits from Banks'', 5) ``Total Loans'', 6) ``Memo: Total Impaired Loans'', 7) ``Derivatives''} on a quarterly basis from 2005-Q4 to 2015-Q3.

Missing data was a big concern for the full dataset 0f 285 institutions. Table \ref{tab:nan-dataset-285} summarizes the percentages of banks affected by missing values at the end of each year for all the balance sheet quantities analyzed. Percentages of missing values in the equity time series range from $3.17\%$ to $93.31\%$ of all quarters among individual banks, with a mean of $(46.86 \pm 29.70)\%$ missing values for the average bank. The time series of interbank data is typically even sketchier, showing $(9.32 \pm 7.58)\%$ more missing values than the equity time series. Finally data on derivatives and impaired loans has the worst coverage of all quantities analyzed.

\begingroup
\setlength{\tabcolsep}{4pt}
\renewcommand{\arraystretch}{1.5}
\begin{table}[!htbp]\index{typefaces!sizes}
  \scriptsize%
  \begin{center}
    \begin{tabular}{llllllllllll}
      \toprule
\textbf{Quantity} &
\begin{tabular}{l} \textbf{2005} \end{tabular} & \begin{tabular}{l} \textbf{2006} \end{tabular} & \begin{tabular}{l} \textbf{2007} \end{tabular} & \begin{tabular}{l} \textbf{2008} \end{tabular} & \begin{tabular}{l} \textbf{2009} \end{tabular} & \begin{tabular}{l} \textbf{2010} \end{tabular} & \begin{tabular}{l} \textbf{2011} \end{tabular} & \begin{tabular}{l} \textbf{2012} \end{tabular} & \begin{tabular}{l} \textbf{2013} \end{tabular} & \begin{tabular}{l} \textbf{2014} \end{tabular} & \begin{tabular}{l} \textbf{2015} \end{tabular}\\
      \midrule
\begin{tabular}{l}Equity \&\\ Assets $\;$ \end{tabular} &  \begin{tabular}{l}  23.16 \end{tabular} &
\begin{tabular}{l}  16.84 \end{tabular} &
\begin{tabular}{l}  15.09 \end{tabular} &
\begin{tabular}{l}  12.28 \end{tabular} &
\begin{tabular}{l}  11.23 \end{tabular} &
\begin{tabular}{l}  8.07 \end{tabular} &
\begin{tabular}{l}  3.86 \end{tabular} &
\begin{tabular}{l}  3.51 \end{tabular} &
\begin{tabular}{l}  3.16 \end{tabular} &
\begin{tabular}{l}  5.96 \end{tabular} &
\begin{tabular}{l}  76.84 \end{tabular}\\
	 \midrule
\begin{tabular}{l} Interbank \& \\ External \\ Assets $\;$ \end{tabular} &
\begin{tabular}{l}  38.25 \end{tabular} &
\begin{tabular}{l}  34.74 \end{tabular} &
\begin{tabular}{l}  34.39 \end{tabular} &
\begin{tabular}{l}  32.28 \end{tabular} &
\begin{tabular}{l}  30.88 \end{tabular} &
\begin{tabular}{l}  28.42 \end{tabular} &
\begin{tabular}{l}  24.91 \end{tabular} &
\begin{tabular}{l}  23.86 \end{tabular} &
\begin{tabular}{l}  23.86 \end{tabular} &
\begin{tabular}{l}  25.61 \end{tabular} &
\begin{tabular}{l}  82.11 \end{tabular}\\
	\midrule
\begin{tabular}{l} Interbank \\ Liabilities $\;$ \end{tabular} &  \begin{tabular}{l}  36.49 \end{tabular} &
\begin{tabular}{l}  31.58 \end{tabular} &
\begin{tabular}{l}  32.98 \end{tabular} &
\begin{tabular}{l}  30.55 \end{tabular} &
\begin{tabular}{l}  28.77 \end{tabular} &
\begin{tabular}{l}  27.37 \end{tabular} &
\begin{tabular}{l}  23.31 \end{tabular} &
\begin{tabular}{l}  24.21 \end{tabular} &
\begin{tabular}{l}  25.26 \end{tabular} &
\begin{tabular}{l}  27.72 \end{tabular} &
\begin{tabular}{l}  81.75 \end{tabular}\\
	 \midrule
\begin{tabular}{l} Derivatives \\ $\;$ \end{tabular} &  \begin{tabular}{l}  59.65 \end{tabular} &
\begin{tabular}{l}  55.09 \end{tabular} &
\begin{tabular}{l}  51.93 \end{tabular} &
\begin{tabular}{l}  49.82 \end{tabular} &
\begin{tabular}{l}  50.18 \end{tabular} &
\begin{tabular}{l}  48.42 \end{tabular} &
\begin{tabular}{l}  41.05 \end{tabular} &
\begin{tabular}{l}  36.84 \end{tabular} &
\begin{tabular}{l}  35.09 \end{tabular} &
\begin{tabular}{l}  38.25 \end{tabular} &
\begin{tabular}{l}  85.96 \end{tabular}\\
	 \midrule
\begin{tabular}{l} Impaired \\ Loans $\;$ \end{tabular} &  \begin{tabular}{l}  68.07 \end{tabular} &
\begin{tabular}{l}  64.91 \end{tabular} &
\begin{tabular}{l}  60.70 \end{tabular} &
\begin{tabular}{l}  59.65 \end{tabular} &
\begin{tabular}{l}  55.79 \end{tabular} &
\begin{tabular}{l}  52.28 \end{tabular} &
\begin{tabular}{l}  44.56 \end{tabular} &
\begin{tabular}{l}  38.95 \end{tabular} &
\begin{tabular}{l}  38.25 \end{tabular} &
\begin{tabular}{l}  41.05 \end{tabular} &
\begin{tabular}{l}  89.12 \end{tabular}\\
\bottomrule
    \end{tabular}
  \end{center}
  \vspace{-0.5cm}
  \caption{Percentages of banks with unavailable values in the full dataset of 285 listed EU banks at each year-end. Breakdown by balance sheet quantity.}
  \label{tab:nan-dataset-285}
\end{table}
\endgroup

Furthermore, despite a positive trend in time that shows higher and more frequent coverage, most banks tend to publish data solely at end or mid-year periods (see Table \ref{tab:reporting-285}).

\begin{table}[!htbp]\index{typefaces!sizes}
  \footnotesize%
  \begin{center}
    \begin{tabular}{lllll}
      \toprule
\textbf{Quantity} &
\begin{tabular}{l} \textbf{Q1} \end{tabular} & 
\begin{tabular}{l} \textbf{Q2} \end{tabular} &
\begin{tabular}{l} \textbf{Q3} \end{tabular} & 
\begin{tabular}{l} \textbf{Q4} \end{tabular}\\
      \midrule
\begin{tabular}{l}Equity \\ Assets $\;$ \end{tabular} &  \begin{tabular}{l}  70.62 \end{tabular} &
\begin{tabular}{l}  35.02 \end{tabular} &
\begin{tabular}{l}  68.33 \end{tabular} &
\begin{tabular}{l}  16.36 \end{tabular}\\
      \midrule
\begin{tabular}{l} Interbank and \\ External \\ Assets $\;$ \end{tabular} &  
\begin{tabular}{l}  72.85 \end{tabular} &
\begin{tabular}{l}  48.96 \end{tabular} &
\begin{tabular}{l}  70.81 \end{tabular} &
\begin{tabular}{l}  34.48 \end{tabular}\\
      \midrule
\begin{tabular}{l} Interbank \\ Liabilities $\;$ \end{tabular} &  \begin{tabular}{l}  73.75 \end{tabular} &
\begin{tabular}{l}  48.48 \end{tabular} &
\begin{tabular}{l}  71.90 \end{tabular} &
\begin{tabular}{l}  33.65 \end{tabular}\\
      \midrule
\begin{tabular}{l} Derivatives $\;$ \end{tabular} &  \begin{tabular}{l}  82.14 \end{tabular} &
\begin{tabular}{l}  61.50 \end{tabular} &
\begin{tabular}{l}  81.02 \end{tabular} &
\begin{tabular}{l}  50.21 \end{tabular}\\
      \midrule
\begin{tabular}{l} Impaired \\ Loans $\;$ \end{tabular} &  \begin{tabular}{l}  87.21 \end{tabular} &
\begin{tabular}{l}  73.72 \end{tabular} &
\begin{tabular}{l}  86.22 \end{tabular} &
\begin{tabular}{l}  55.76 \end{tabular}\\
\bottomrule
    \end{tabular}
  \end{center}
  \vspace{-0.5cm}
  \caption{Percentages of banks failing to report balance sheet quantitied in the full dataset of 285 listed EU banks. Breakdown by balance sheet quantity and quarter. Percentages at fixed quarter refer to averages on specified quarter of all years in the dataset.}
  \label{tab:reporting-285}
\end{table}

Therefore we focus on a subset of the original dataset and select the top 50 banks by total assets in 2013-Q4, so as to maximize data coverage (bigger institutions tend to report more data, more often) and obtain a representative sample of the EU banking system. 
Table \ref{tab:nan-dataset-50} is analogous to Table \ref{tab:nan-dataset-285} and shows the clear gain in data coverage for this restriction of the dataset.

\begingroup
\setlength{\tabcolsep}{4pt}
\renewcommand{\arraystretch}{1.5}
\begin{table}[!htbp]\index{typefaces!sizes}
  \scriptsize%
  \begin{center}
    \begin{tabular}{llllllllllll}
      \toprule
\textbf{Quantity} &
\begin{tabular}{l} \textbf{2005} \end{tabular} & \begin{tabular}{l} \textbf{2006} \end{tabular} & \begin{tabular}{l} \textbf{2007} \end{tabular} & \begin{tabular}{l} \textbf{2008} \end{tabular} & \begin{tabular}{l} \textbf{2009} \end{tabular} & \begin{tabular}{l} \textbf{2010} \end{tabular} & \begin{tabular}{l} \textbf{2011} \end{tabular} & \begin{tabular}{l} \textbf{2012} \end{tabular} & \begin{tabular}{l} \textbf{2013} \end{tabular} & \begin{tabular}{l} \textbf{2014} \end{tabular} & \begin{tabular}{l} \textbf{2015} \end{tabular}\\
      \midrule
\begin{tabular}{l} Interbank \\ Liabilities $\;$ \end{tabular} &
\begin{tabular}{l}  1.40 \end{tabular} &
\begin{tabular}{l}  0.35 \end{tabular} &
\begin{tabular}{l}  0.00 \end{tabular} &
\begin{tabular}{l}  0.35 \end{tabular} &
\begin{tabular}{l}  0.35 \end{tabular} &
\begin{tabular}{l}  0.35 \end{tabular} &
\begin{tabular}{l}  0.35 \end{tabular} &
\begin{tabular}{l}  0.35 \end{tabular} &
\begin{tabular}{l}  0.35 \end{tabular} &
\begin{tabular}{l}  1.05 \end{tabular} &
\begin{tabular}{l}  9.47 \end{tabular}\\
	\midrule
\begin{tabular}{l} Derivatives$\;$ \end{tabular} &  \begin{tabular}{l}  1.40 \end{tabular} &
\begin{tabular}{l}  0.35 \end{tabular} &
\begin{tabular}{l}  0.00 \end{tabular} &
\begin{tabular}{l}  0.00 \end{tabular} &
\begin{tabular}{l}  0.00 \end{tabular} &
\begin{tabular}{l}  0.00 \end{tabular} &
\begin{tabular}{l}  0.00 \end{tabular} &
\begin{tabular}{l}  0.00 \end{tabular} &
\begin{tabular}{l}  0.00 \end{tabular} &
\begin{tabular}{l}  0.35 \end{tabular} &
\begin{tabular}{l}  10.53 \end{tabular}\\
	 \midrule
\begin{tabular}{l} Impaired \\ Loans $\;$ \end{tabular} &  \begin{tabular}{l}  3.86 \end{tabular} &
\begin{tabular}{l}  1.75 \end{tabular} &
\begin{tabular}{l}  0.70 \end{tabular} &
\begin{tabular}{l}  0.70 \end{tabular} &
\begin{tabular}{l}  0.70 \end{tabular} &
\begin{tabular}{l}  1.05 \end{tabular} &
\begin{tabular}{l}  0.35 \end{tabular} &
\begin{tabular}{l}  0.35 \end{tabular} &
\begin{tabular}{l}  0.70 \end{tabular} &
\begin{tabular}{l}  0.70 \end{tabular} &
\begin{tabular}{l}  1.93 \end{tabular}\\
	 \midrule
\begin{tabular}{l} All other \\ quantities $\;$ \end{tabular} &  \begin{tabular}{l}  1.40 \end{tabular} &
\begin{tabular}{l}  $\leq$ 0.35 \end{tabular} &
\begin{tabular}{l}  0.00 \end{tabular} &
\begin{tabular}{l}  0.00 \end{tabular} &
\begin{tabular}{l}  0.00 \end{tabular} &
\begin{tabular}{l}  0.00 \end{tabular} &
\begin{tabular}{l}  0.00 \end{tabular} &
\begin{tabular}{l}  0.00 \end{tabular} &
\begin{tabular}{l}  0.00 \end{tabular} &
\begin{tabular}{l}  $\leq$ 0.35 \end{tabular} &
\begin{tabular}{l}  $\leq$ 8.77 \end{tabular}\\
\bottomrule
    \end{tabular}
  \end{center}
  \vspace{-0.5cm}
  \caption{Percentages of banks with unavailable values in the restricted dataset of top 50 listed EU banks by total assets. Breakdown by balance sheet quantity. Percentage refer to end-of-year data.}
\label{tab:nan-dataset-50}
\end{table}
\endgroup

Bigger banks indeed tend to report their balance sheet data more frequently and more consistently, thus restricting to the subset of the top 50 banks by assets leads to a dramatic improvement in data quality. Table \ref{tab:coverage-dataset-50} shows the percentage of basic balance sheet quantities covered by our selection of top 50 banks with respect to the original dataset of 285 banks. Indeed this subsample is shown to be representantive of the whole EU banking system in absolute terms, since most of equity and assets is concentrated in these bigger institutions.

\begingroup
\setlength{\tabcolsep}{4pt}
\renewcommand{\arraystretch}{1.5}
\begin{table}[!htbp]\index{typefaces!sizes}
  \scriptsize%
  \begin{center}
    \begin{tabular}{llllllllllll}
      \toprule
\textbf{Quantity} &
\begin{tabular}{l} \textbf{2005} \end{tabular} & \begin{tabular}{l} \textbf{2006} \end{tabular} & \begin{tabular}{l} \textbf{2007} \end{tabular} & \begin{tabular}{l} \textbf{2008} \end{tabular} & \begin{tabular}{l} \textbf{2009} \end{tabular} & \begin{tabular}{l} \textbf{2010} \end{tabular} & \begin{tabular}{l} \textbf{2011} \end{tabular} & \begin{tabular}{l} \textbf{2012} \end{tabular} & \begin{tabular}{l} \textbf{2013} \end{tabular} & \begin{tabular}{l} \textbf{2014} \end{tabular} & \begin{tabular}{l} \textbf{2015} \end{tabular}\\
      \midrule
\begin{tabular}{l} Equity \end{tabular} &
\begin{tabular}{l}  82.45 \end{tabular} &
\begin{tabular}{l}  82.46 \end{tabular} &
\begin{tabular}{l}  82.09 \end{tabular} &
\begin{tabular}{l}  83.57 \end{tabular} &
\begin{tabular}{l}  84.07 \end{tabular} &
\begin{tabular}{l}  81.38 \end{tabular} &
\begin{tabular}{l}  80.53 \end{tabular} &
\begin{tabular}{l}  82.29 \end{tabular} &
\begin{tabular}{l}  80.16 \end{tabular} &
\begin{tabular}{l}  81.03 \end{tabular} &
\begin{tabular}{l}  89.97 \end{tabular}\\
	 \midrule
\begin{tabular}{l} Assets \end{tabular} &
\begin{tabular}{l}  94.23 \end{tabular} &
\begin{tabular}{l}  93.99 \end{tabular} &
\begin{tabular}{l}  94.52 \end{tabular} &
\begin{tabular}{l}  94.00 \end{tabular} &
\begin{tabular}{l}  91.78 \end{tabular} &
\begin{tabular}{l}  89.83 \end{tabular} &
\begin{tabular}{l}  89.30 \end{tabular} &
\begin{tabular}{l}  89.09 \end{tabular} &
\begin{tabular}{l}  88.38 \end{tabular} &
\begin{tabular}{l}  89.50 \end{tabular} &
\begin{tabular}{l}  92.89 \end{tabular}\\
	\midrule
\begin{tabular}{l} Interbank \\ Assets \end{tabular} &  \begin{tabular}{l}  94.42 \end{tabular} &
\begin{tabular}{l}  93.68 \end{tabular} &
\begin{tabular}{l}  93.96 \end{tabular} &
\begin{tabular}{l}  87.31 \end{tabular} &
\begin{tabular}{l}  82.37 \end{tabular} &
\begin{tabular}{l}  85.98 \end{tabular} &
\begin{tabular}{l}  82.42 \end{tabular} &
\begin{tabular}{l}  80.50 \end{tabular} &
\begin{tabular}{l}  84.19 \end{tabular} &
\begin{tabular}{l}  86.02 \end{tabular} &
\begin{tabular}{l}  72.68 \end{tabular}\\
	 \midrule
\begin{tabular}{l} Interbank \\ Liabilities \end{tabular} &  \begin{tabular}{l}  95.13 \end{tabular} &
\begin{tabular}{l}  94.12 \end{tabular} &
\begin{tabular}{l}  93.81 \end{tabular} &
\begin{tabular}{l}  90.55 \end{tabular} &
\begin{tabular}{l}  85.59 \end{tabular} &
\begin{tabular}{l}  82.76 \end{tabular} &
\begin{tabular}{l}  84.67 \end{tabular} &
\begin{tabular}{l}  83.81 \end{tabular} &
\begin{tabular}{l}  78.80 \end{tabular} &
\begin{tabular}{l}  83.03 \end{tabular} &
\begin{tabular}{l}  93.99 \end{tabular}\\
	 \midrule
\begin{tabular}{l} Derivatives \end{tabular} &  \begin{tabular}{l}  99.36 \end{tabular} &
\begin{tabular}{l}  99.36 \end{tabular} &
\begin{tabular}{l}  99.44 \end{tabular} &
\begin{tabular}{l}  99.40 \end{tabular} &
\begin{tabular}{l}  98.90 \end{tabular} &
\begin{tabular}{l}  97.70 \end{tabular} &
\begin{tabular}{l}  97.52 \end{tabular} &
\begin{tabular}{l}  96.92 \end{tabular} &
\begin{tabular}{l}  97.08 \end{tabular} &
\begin{tabular}{l}  97.69 \end{tabular} &
\begin{tabular}{l}  99.44 \end{tabular}\\
	 \midrule
\begin{tabular}{l} Impaired \\ Loans \end{tabular} &  \begin{tabular}{l}  89.80 \end{tabular} &
\begin{tabular}{l}  92.80 \end{tabular} &
\begin{tabular}{l}  94.39 \end{tabular} &
\begin{tabular}{l}  93.68 \end{tabular} &
\begin{tabular}{l}  92.31 \end{tabular} &
\begin{tabular}{l}  91.37 \end{tabular} &
\begin{tabular}{l}  89.40 \end{tabular} &
\begin{tabular}{l}  87.32 \end{tabular} &
\begin{tabular}{l}  86.16 \end{tabular} &
\begin{tabular}{l}  87.94 \end{tabular} &
\begin{tabular}{l}  90.65 \end{tabular}\\
	 \midrule
\begin{tabular}{l} External \\ Assets \end{tabular} &  \begin{tabular}{l}  95.77 \end{tabular} &
\begin{tabular}{l}  95.96 \end{tabular} &
\begin{tabular}{l}  95.90 \end{tabular} &
\begin{tabular}{l}  95.44 \end{tabular} &
\begin{tabular}{l}  93.82 \end{tabular} &
\begin{tabular}{l}  91.50 \end{tabular} &
\begin{tabular}{l}  91.32 \end{tabular} &
\begin{tabular}{l}  90.75 \end{tabular} &
\begin{tabular}{l}  90.03 \end{tabular} &
\begin{tabular}{l}  91.18 \end{tabular} &
\begin{tabular}{l}  93.84 \end{tabular}\\
\bottomrule
    \end{tabular}
  \end{center}
  \vspace{-0.5cm}
  \caption{Coverage of balance sheet quantities by the restricted dataset of top 50 listed EU banks with respect to the original dataset of 285 banks. Breakdown by balance sheet quantity. Percentages refer to end-of-year data.}
\label{tab:coverage-dataset-50}
\end{table}
\endgroup

For this choice of dataset the time series of total assets is complete. For some problematic banks, the equity and total loans series present an average of $(23.15\pm22.66)\%$ missing data points but never more than three of them consecutively therefore they have been completely reconstructed by simple linear interpolation.

In the case of missing interbank assets data we performed an estimation via linear interpolation on interbank leverages available for other years. The choice of interpolating on interbank leverages instead of directly interpolating on interbank assets is due to the fact that the former are more stable in time, being ideally free of any trend, and can be easily obtained from the equity time series, which is complete.

Analogously, for missing interbank liabilities we interpolated on the fraction of interbank liabilities to total liabilities, which is also more stable in time than the corresponding time series on interbank liabilities and allows to fully exploit the completely reconstructed liabilities time series (difference between assets and equity).

Finally, missing impaired loans have been reconstructed by linearly interpolating on the fraction of impaired loans to total loans.

\subsection{Network reconstruction}
\label{subsec:si-network-reconstruction}

Detailed bilateral exposures between individual institutions are typically confidential. Several works on national interbank networks have used real bilateral data collected by national central banks, but at the larger EU-level there is no comprehensive dataset available.

At each point in time we reconstruct $1,000$ interbank networks from aggregate data of the individual banks in our dataset. For the reconstruction of the adjacency matrix we follow the ``fitness model'' outlined in \cite{demasi2006fitness}, \cite{musmeci2012bootstrapping} and \cite{montagna2016contagion}, which allows to replicate a core-periphery topology typical of interbank markets. The weighted adjacency matrix is then obtained by calculating weights using a proportional fitting algorithm introduced in \cite{battiston2016leveraging}. 

More in detail, the procedure is as follows:
\begin{itemize}
\item \textbf{Total exposure rebalancing.} Our dataset represents only a subset of the entire interbank market. We notice that total interbank assets, $A^b = \sum A_i^b$, are systematically lower than total interbank liabilities, $L^b = \sum L_i^b$ (ranging from a minimum of $59.4\%$ to a maximum of $92.2\%$). This implies that the top 50 EU listed banks are net borrowers. We take a conservative stand and assume that the volume of total lending in the network is the minimum between the two, i.e. total interbank assets.
\item \textbf{Link assignment} In agreement with the fitness model, define $x_i = \frac{1}{2}\big(\frac{A_i^b}{A^b} + \frac{L_i^b}{L^b}\big)$ as the ``fitness'' of bank $i$, which is a proxy for the importance of the bank in the interbank market.
We can then estimate the probability that a bilateral exposure between bank $i$ and bank $j$ exists as $ p_{ij} = \frac{z x_i x_j}{1 + z x_i x_j}$. The parameter $z$ is free and is chosen in such a way that the average density is $20\%$, compatibly with the fact that our top 50 EU listed banks represent the core of the interbank market and are highly interlinked\footnote{We have performed robustness checks for this particular choice of density and noticed negligible changes in results for lower densities}. Subsequently, $1,000$ realizations of this prior matrix are sampled, yielding $1,000$ adjacency matrices.
\item \textbf{Weight assignment} The last step is that of assigning weights to each matrix. Following the methodology presented in \cite{battiston2016leveraging} we impose the constraint that the sum of exposures of each bank equals its total interbank assets, $A_i^b$. We implement an iterative proportional fitting algorithm on the interbank exposure matrix.
In a nutshell, we wish to estimate the matrix $\pi_{ij} = A_{ij}^b / A^b$, which is simply the matrix $A_{ij}^b$ normalized on the sum of all entries. Each iteration of the fitting procedure consists of two steps:
\begin{enumerate}
\item $\hat{\pi}_{ij}' = \frac{\hat{\pi}_{ij}}{\sum_j \hat{\pi}_{ij}} \frac{A_i^b}{A^b}$
\item $\hat{\pi}_{ij}'' = \frac{\hat{\pi}'_{ij}}{\sum_i \hat{\pi}'_{ij}} \frac{L_i^b}{L^b}$
\end{enumerate}
The procedure stops as soon as $\sum_j \hat{\pi}_{ij} - A_i^b /A^b$ and $\sum_i \hat{\pi}_{ij} - L_i^b /L^b$ are below $1\%$.
Finally the estimated weighted adjacency matrix is obtained as $\pi_{ij} \cdot A^b$.
\end{itemize}

\clearpage

\bibliographystyle{apalike}
\bibliography{references}
\end{document}